\documentclass{LMCS}

\def\dOi{11(1:1)2015}
\lmcsheading%
{\dOi}
{1--32}
{}
{}
{Apr.~\phantom03, 2011}
{Feb.~10, 2015}
{}

\ACMCCS{[{\bf Theory of computation}]:  Semantics and
  Reasoning---Program semantics---Operational semantics; Semantics and
  Reasoning---Program reasoning; Logic---Hoare logic}

\newcommand{\REMOVED}[1]{}

\usepackage{hyperref}
\usepackage[anythingbreaks]{breakurl}

\usepackage{amssymb}
\usepackage{proof}
\usepackage{stmaryrd}

\newcommand{\whilefun}{\textrm{while-fun}}
\newcommand{\whilerel}{\textrm{while-rel}}

\newcommand{\IH}{\mathrm{IH}}
\newcommand{\Hyp}{\mathrm{H}}

\newcommand{\twoline}[2]{\begin{array}{l}#1 \\ \quad #2 \end{array}}

\newcommand{\detach}[3]{\hspace*{2mm} \deduce{\line(#1,6){#2}}{\makebox[0cm]
{#3}} \hspace*{2mm}}

\newcommand{\Skip}{\ensuremath{\mathsf{skip}}}
\newcommand{\Assign}[2]{\ensuremath{#1 := #2}}
\newcommand{\Seq}[2]{\ensuremath{#1;#2}}
\newcommand{\Ifthenelse}[3]
{\ensuremath{\mathsf{if~} #1 \mathsf{~then~} #2 \mathsf{~else~} #3}}
\newcommand{\While}[2]{\ensuremath{\mathsf{while~} #1 \mathsf{~do~} #2}}
\newcommand{\Abort}{\mathsf{abort}}
\newcommand{\true}{\ensuremath{\mathsf{true}}}
\newcommand{\false}{\ensuremath{\mathsf{false}}}

\newcommand{\eval}[2]{\ensuremath{\llbracket #1 \rrbracket #2}}
\newcommand{\istrue}[2]{\ensuremath{#2 \models #1}}
\newcommand{\isfalse}[2]{\ensuremath{#2 \not\models #1}}
\newcommand{\update}[3]{\ensuremath{#1[#2\mapsto #3]}}

\newcommand{\st}{\sigma}
\newcommand{\state}{\mathit{state}}
\newcommand{\tr}{\ensuremath{\tau}}
\newcommand{\trace}{\mathit{trace}}
\newcommand{\sglt}[1]{\langle#1\rangle}
\newcommand{\cons}[2]{\ensuremath{#1::#2}}
\newcommand{\bism}[2]{\ensuremath{#1 \approx #2}}

\newcommand{\hd}{\mathit{hd}}

\newcommand{\concat}{\mathbin{++}}

\newcommand{\exec}[3]{(#1,#2) \Rightarrow #3}
\newcommand{\execseq}[3]{(#1,#2) \stackrel{*}{\Rightarrow} #3}
\newcommand{\execnoargs}{\Rightarrow}
\newcommand{\execseqnoargs}{\stackrel{*}{\Rightarrow}}

\newcommand{\execind}[3]{(#1,#2) \Rightarrow^\mathrm{ind} #3}

\newcommand{\duplast}{\mathit{duplast}}

\renewcommand{\split}[4]{\mathit{split} ~#1 ~#2 ~#3 ~#4}

\newcommand{\semax}[3]{\{#1\}~#2~\{#3\}}
\newcommand{\Follows}[2]{#1 \mathbin{\ast\ast} #2}
\newcommand{\chop}{\mathbin{\ast\ast}}
\newcommand{\follows}[3]{#3 \models_{#2} #1}

\renewcommand{\conj}[2]{#1 \wedge #2}
\newcommand{\disj}[2]{#1 \vee #2}

\newcommand{\ceil}[2]{\lceil #1 \rceil #2}
\newcommand{\after}[2]{#1 \triangleleft #2}

\newcommand{\sem}[1]{\llbracket #1 \rrbracket}
\newcommand{\satisfy}[2]{#1 \models #2}
\newcommand{\imp}{\to}
\newcommand{\entails}{\models}
\newcommand{\logequ}{\Leftrightarrow}

\newcommand{\finite}{\mathit{finite}}
\newcommand{\infinite}{\mathit{infinite}}
\newcommand{\cofinally}{{\it cofinally}}

\newcommand{\eventually}{\mathit{eventually}}

\newcommand{\upstream}{\mathit{up}}

\newcommand{\stay}[1]{\langle #1 \rangle^*}

\newcommand{\red}[2]{#1 \stackrel{*}{\rightsquigarrow} #2}

\newcommand{\rar}{\rightarrow}

\newcommand{\dup}[1]{\langle #1 \rangle^2}
\newcommand{\rep}[1]{#1^\dagger}

\renewcommand{\conv}[2]{#1 \mathbin{\downarrow} #2}

\renewcommand{\div}[1]{#1^{\shortuparrow}}

\newcommand{\Last}[1]{\mathit{Last} ~#1}
\newcommand{\Lastnoargs}{\mathit{Last}}

\newcommand{\Inv}[3]{\mathit{Inv}(#1, #2, #3)}

\newcommand{\len}[1]{\mathit{len}\, #1}

\renewcommand{\sp}[2]{\mathit{sp}(#1, #2)}

\newcommand{\midp}{\mathit{midp}}

\begin{document}

\title[A Hoare Logic for the Coinductive Trace-Based Big-Step
  Semantics of While]{A Hoare Logic for the Coinductive Trace-Based Big-Step
  Semantics of While\rsuper*}

\author[K.~Nakata]{Keiko Nakata} 
\address{Institute of Cybernetics at Tallinn
  University of Technology, Akadeemia tee 21, 12618 Tallinn, Estonia}
\email{\{keiko,tarmo\}@cs.ioc.ee} 
\thanks{The authors were supported by the EU
  FP6 IST integrated project no.\ 15905 (MOBIUS), the ERDF funded
  Estonian CoE project EXCS, the Estonian Ministry of Education and
  Research target-financed themes no.\ 0322709s06 and 0140007s12, 
  and the Estonian Science Foundation grants no.\ 6940 and 9398.}

\author[T.~Uustalu]{Tarmo Uustalu}
\address{\vspace{-18 pt}}

\titlecomment{{\lsuper*}This article is a revised and expanded version of the
  ESOP 2010 conference paper~\cite{esop10}.}

\begin{abstract}
  In search for a foundational framework for reasoning about
  observable behavior of programs that may not terminate, we have
  previously devised a trace-based big-step semantics for While. In
  this semantics, both traces and evaluation (relating initial states
  of program runs to traces they produce) are defined coinductively.
  On terminating runs, this semantics agrees with the standard inductive
  state-based semantics. Here we present a Hoare logic counterpart of
  our coinductive trace-based semantics and prove it sound and
  complete.  Our logic subsumes the standard partial-correctness
  state-based Hoare logic as well as the total-correctness variation:
  they are embeddable. In the converse direction, projections can 
  be constructed: a derivation of a Hoare triple in our trace-based
  logic can be translated into a derivation in the state-based logic
  of a translated, weaker Hoare triple.  Since we work with a
  constructive underlying logic, the range of program
  properties we can reason about has a fine
  structure; in particular, we can distinguish
  between termination and nondivergence, e.g., unbounded classically
  total search fails to be terminating, but is nonetheless
  nondivergent. Our metatheory is entirely constructive as well, and
  we have formalized it in Coq.
\end{abstract}

\maketitle


\section{Introduction}

Standard big-step semantics and Hoare logics do
not support reasoning about nonterminating runs of programs.
Essentially, they ignore them. But of course nonterminating runs are
important. 
Not only need we often program a partial function whose
domain of definedness we cannot decide or is undecidable, e.g., an
interpreter, but we also have to 
program functions that are
inherently partial. In programming with interactive
input/output, for example, diverging runs are often what we
really want. 

In search for a foundational framework for reasoning about possibly
nonterminating programs constructively (intuitionistically) and intrigued by attempts in this direction in
the literature, we have previously devised a big-step semantics for
While based on traces \cite{NU:trabco}. 
In this semantics, traces are possibly infinite
sequences of states that a program run goes through. They are defined
coinductively, as is the evaluation relation, relating initial states
of program runs to traces they produce. On terminating runs, this
nonstandard semantics agrees with the standard, inductive state-based
big-step semantics.

In this paper, we put forward a Hoare logic to match this big-step
semantics. In this new trace-based logic, program runs are reasoned
about in terms of assertions on states and traces. More precisely, our
Hoare triple $\semax{U}{s}{P}$ is given by a statement $s$, a state
assertion $U$ (a condition on the initial state of a run of $s$) and a
trace assertion $P$ (a condition on the trace produced by the run).
In the presentation we have chosen for this paper, assertions are
nothing but predicates expressible in the meta-logic, i.e., we do not
confine ourselves to a particular language of state and trace
assertions. Nonetheless, we do not want to downplay the question of
what makes a good assertion language for traces. We are after a set of
connectives that allows for a concise formulation of a sound and
complete Hoare logic over state and trace predicates and logical
entailment as given by the constructive meta-logic. We adopt a
solution that is reminiscent of interval temporal logic
\cite{Mos:temlrh,HZM:harsbt}, with a chop-connective. This gives us a set of
connectives that is Spartan in terms of convenience of expression, but
suffices for our meta-theoretical study. Our logic is intended
foundational framework into which more specialized and more applied
logics with more limited assertion languages can be embedded.

Besides being deterministic, the While language is also total as soon
as we accept that traces of program runs can be infinite. This allows
our logic to conservatively extend both the standard, state-based
partial-correctness Hoare logic as well as the state-based
total-correctness Hoare logic. On the level of derivability alone this
can be proved semantically by going through the soundness and
completeness results. But we go one step further: we show that
derivations in these two state-based logics are directly transformable
into derivations in our logic, yielding embeddings on the level of
derivations, not just mere derivability. The transformations are
relatively straightforward and do not require invention of new
invariants or variants, demonstrating that our logic incurs no undue
proof burden in comparison to the standard Hoare logics. In the
converse direction, we can project derivations in our trace-based
logic into derivations in the state-based logics: a derivation of a
Hoare triple in the trace-based logic is translated into a derivation
in the state-based logics with a translated, weaker postcondition.

However, the power of our logic goes beyond that of the state-based
partial-correctness and total-correctness Hoare logics. The assertions
have access to traces. As suggested by the similarity of our (open)
assertion language to the that of interval temporal logic, this allows
us to specify liveness properties of diverging runs.  We will
demonstrate this extra expressiveness of our logic by a series of
examples. Also, interpreted into a constructive underlying logic, our
assertion language becomes quite discerning. In particular we can
distinguish between termination and nondivergence, e.g., unbounded
classically total (constructively nonpartial) search fails to be
terminating, but is nonetheless nondivergent.

We do not discuss this in the paper, but our logic can be
adjusted to deal with exceptions and nondeterminism. 

The paper is organized as follows. In Section~\ref{sec:sem}, we
present our trace-based big-step semantics.  In
Section~\ref{sec:hoare}, we proceed to the question of a corresponding
Hoare logic. We explain our design considerations and then present our
Hoare logic and the soundness and completeness proofs.  In
Section~\ref{sec:Hoare}, we show the embeddings of the state-based
partial-correctness and total-correctness Hoare logics into our logic
and the projections back.  In Section~\ref{sec:examples}, we consider
examples.  In Section~\ref{sec:related}, we discuss the related work,
to conclude in Section~\ref{sec:concl}.  

We have formalized the development fully constructively in Coq version
8.1pl3 using the \textsf{Ssreflect} syntax extension library. The Coq
development is available at
\url{http://cs.ioc.ee/~keiko/code/abyss.tgz}.

Both the paper and the accompanying Coq code use coinductive types,
corecursion and coinduction extensively. For an introduction, we can
refer the reader to the exposition of Bertot and Cast\'eran
\cite[Ch.~13]{coqart}. In the paper, we have sought to abstract over
the more bureaucratic aspects involved in the Coq formalization (e.g.,
working with Coq's restricted guardedness condition on cofix
definitions, i.e., definitions by corecursion, proofs by coinduction).


\section{Big-step semantics}
\label{sec:sem}

We start with our big-step semantics. This is defined in terms of
states and traces. The notion of a state is standard. A state $\st \in
\state$ is an assignment of integer values to the variables. Traces
$\tr \in \trace$ are defined coinductively by the rules\footnote{We
  mark coinductive definitions by double horizontal rules.}
\[
\infer={\sglt{\st} \in \trace}{
}
\quad
\infer={\cons{\st}{\tr} \in \trace}{
  \tr \in \trace
}
\]
so a trace is a non-empty colist (possibly infinite sequence) of
states. We also define (strong) bisimilarity of two traces, $\bism{\tr}{\tr'}$,
coinductively by
\[
\infer={\bism{\sglt{\st}}{\sglt{\st}}}{
}
\quad
\infer={\bism{\cons{\st}{\tr}}{\cons{\st}{\tr'}}}{
  \bism{\tr}{\tr'}
}
\]
Bisimilarity is straightforwardly seen to be an equivalence. We think
of bisimilar traces as equal, i.e., type-theoretically we treat traces
as a setoid with bisimilarity as the equivalence
relation.\footnote{Classically, strong bisimilarity is equality. But we
  work in an intensional type theory where strong bisimilarity of
  colists is weaker than equality (just as equality of two functions
  on all arguments is weaker than equality of these two
  functions).}\footnote{In particular, we ``pattern-match'' traces
  using bisimilarity: any trace $\tr$ is bisimilar to either a trace
  of the form of $\sglt{\st}$ or one of the form
  $\cons{\st}{\tr'}$.} Accordingly, we have to make sure that all
functions and predicates we define on traces are setoid functions and
predicates (i.e., insensitive to bisimilarity). We define the initial
state $\hd~\tr$ of a trace $\tr$ by case distinction by
$\hd~\sglt{\st} = \st, \hd~(\cons{\st}{\tr}) = \st$. The function
$\hd$ is a setoid function. We also define finiteness of a trace (with
a particular final state) and infiniteness of a trace inductively
resp.\ coinductively by
\[
\infer{\conv{\sglt{\st}}{\st}}{
}
\quad
\infer{\conv{\cons{\st}{\tr}}{\st'}}{ 
  \conv{\tr}{\st'}
}
\hspace*{1.5cm}
\infer={\div{(\cons{\st}{\tr})}}{
  \div{\tr}
}
\]
Finiteness and infiniteness are setoid predicates. 
It should be noticed that infiniteness is defined 
positively, not as negation of finiteness. 
Constructively, it
is not the case that $\forall \tr.\, 
(\exists \st.\, \conv{\tr}{\st}) \vee \div{\tr}$, 
which amounts to asserting that finiteness is decidable. 
In particular, 
$\forall \tr.\,(\neg\, \exists \st.\, \conv{\tr}{\st})
\to \div{\tr}$
is constructively provable, but 
$\forall \tr.\, \neg\, \div{\tr} \to \exists \st.\, \conv{\tr}{\st}$ is not. 



The statements of the While language are given by the following
grammar where $x$ ranges over (integer) variables and $e$ over
(arithmetic) expressions built over variables.
\[
s ::= \Assign{x}{e} \mid \Skip \mid \Seq{s_0}{s_1}
    \mid \Ifthenelse{e}{s_t}{s_f} \mid \While{e}{s_t}
\]
The integer value of an expression $e$ in a state $\st$ is denoted
$\sem{e}{\st}$. We also interpret expressions as booleans; $\st
\models e$ stands for $e$ being true in $\st$. 
Evaluation $\exec{s}{\st}{\tr}$, expressing that running a statement
$s$ from a state $\st$ produces a trace $\tr$, is defined
coinductively by the rules in Figure~\ref{fig:sem}. The rules for
sequence and while implement the necessary sequencing with the help of
extended evaluation $\execseq{s}{\tr}{\tr'}$, 
also defined coinductively, as the coinductive prefix
closure of evaluation:
$\execseq{s}{\tr}{\tr'}$ expresses that 
running
a statement $s$ from the last state (if it exists) of an already
accumulated trace $\tr$ results in a total trace $\tr'$.  

A remarkable feature of the definition of $\execseq{s}{\tr}{\tr'}$ 
is that it does not hinge on deciding whether the trace $\tr$ 
is finite or not, which
is constructively impossible. A proof of $\execseq{s}{\tr}{\tr'}$ 
simply traverses the already accumulated trace $\tr$: if the last
element is hit, which is the case when $\tr$ is finite, 
then the statement is run, otherwise the traversal goes on
forever.

We look closer at the sequence rule. We want to conclude that
$\exec{\Seq{s_0}{s_1}}{\st}{\tr'}$ from the premise
$\exec{s_0}{\st}{\tr}$.  Classically, either the run of $s_0$ terminates, i.e.,
$\conv{\tr}{\st'}$ for some $\st'$, or it diverges, i.e., $\div{\tr}$.
In the first case, we would like to additionally use that $\tr$ is a
finite prefix of $\tr'$ and that $\exec{s_1}{\st'}{\tr''}$, where
$\tr''$ is the rest of $\tr'$. In the second case, it should be case that
$\bism{\tr}{\tr'}$. In both cases, the desirable condition is
equivalent to $\execseq{s_1}{\tr}{\tr'}$, which is the second premise
of our rule.  The use of extended evaluation, defined as the
coinductive (rather than inductive) prefix closure of evaluation,
allows us to avoid the need to decide whether the run of $s_0$ terminates or
not.


\begin{figure}
\[
\begin{array}{c}
\infer={
  \exec{\Assign{x}{e}}{\st}{\cons{\st}{\sglt{\update{\st}{x}{\eval{e}{\st}}}}}
}{}
\quad
\infer={
  \exec{\Skip}{\st}{\sglt{\st}}
}{}
\quad
\infer={
  \exec{\Seq{s_0}{s_1}}{\st}{\tr'}
}{
  \exec{s_0}{\st}{\tr}
  &\execseq{s_1}{\tr}{\tr'}
}
\\[1ex]
\infer={
  \exec{\Ifthenelse{e}{s_t}{s_f}}{\st}{\tr}
}{
  \istrue{e}{\st}
  &\execseq{s_t}{\cons{\st}{\sglt{\st}}}{\tr}
}
\quad
\infer={
  \exec{\Ifthenelse{e}{s_t}{s_f}}{\st}{\tr}
}{
  \isfalse{e}{\st}
  &\execseq{s_f}{\cons{\st}{\sglt{\st}}}{\tr}
}
\\[1ex]
\infer={
  \exec{\While{e}{s_t}}{\st}{\tr'}
}{
  \istrue{e}{\st}
  &\execseq{s_t}{\cons{\st}{\sglt{\st}}}{\tr}
  &\execseq{\While{e}{s_t}}{\tr}{\tr'}
}
\quad
\infer={
  \exec{\While{e}{s_t}}{\st}{\cons{\st}{\sglt{\st}}}
}{
  \isfalse{e}{\st}
}
\\[2ex]
\infer={
  \execseq{s}{\sglt{\st}}{\tr}
}{
  \exec{s}{\st}{\tr}
}
\quad
\infer={
  \execseq{s}{\cons{\st}{\tr}}{\cons{\st}{\tr'}}
}{
  \execseq{s}{\tr}{\tr'}
}
\end{array}
\]
\caption{Big-step semantics}\label{fig:bigrel}
\label{fig:sem}
\end{figure}

Evaluation is a setoid predicate.

\begin{prop}[\cite{NU:trabco}]
  For any $s$, $\st$, $\tr$ and $\tr'$, if $\exec{s}{\st}{\tr}$ and
  $\bism{\tr}{\tr'}$, then $\exec{s}{\st}{\tr'}$.
\end{prop}

The trace produced from a state begins with this state.

\begin{prop}
  For any $s$, $\st$ and $\tr$, if $\exec{s}{\st}{\tr}$, then $\hd~
  \tr = \st$.
\end{prop}

Moreover, for While, evaluation is
deterministic (up to bisimilarity, as is appropriate for our notion of
trace equality).

\begin{prop}[\cite{NU:trabco}]\label{prop:determ}
  For any $s$, $\st$, $\tr$ and $\tr'$, if $\exec{s}{\st}{\tr}$ and
  $\exec{s}{\st}{\tr'}$, then $\bism{\tr}{\tr'}$.
\end{prop}

And it is also total.

\begin{prop}[\cite{NU:trabco}]\label{prop:total}
  For any $s$ and $\st$, there exists $\tr$ such that $\exec{s}{\st}{\tr}$.
\end{prop}

In our definition, we have made a choice as regards to what grows the
trace of a run. We have decided that assignments and testing of guards
of if- and while-statements augment the trace by a state (but $\Skip$
does not), 
e.g., we have $\exec{x
  :=17}{\st}{\cons{\st}{\sglt{\update{\st}{x}{17}}}}$,
$\exec{\While{\false}{\Skip}}{\st}{\cons{\st}{\sglt{\st}}}$ and
$\exec{\While{\true}{\Skip}}{\st}{\cons{\st}{\cons{\st}{\cons{\st}{\ldots}}}}$.

This is good for several reasons. 
First, $\Skip$ becomes the unit of sequential composition, 
i.e., the semantics does not
distinguish $s$, $\Seq{\Skip}{s}$ and $\Seq{s}{\Skip}$.  
Second, we get a notion of small steps
that fully agrees with a very natural coinductive trace-based
small-step semantics arising as a straightforward variation of the
textbook inductive state-based small-step semantics.
The third and most
important outcome is that any while-loop always progresses, because
testing of the guard is a small step.  
For instance, in our semantics $\While{\true}{\Skip}$ can only derive 
$\exec{\While{\true}{\Skip}}{\st}{\cons{\st}{\cons{\st}{\cons{\st}{\ldots}}}}$
(up to bisimilarity). 
As we discuss below, giving up insisting on progress in terms of growing
the trace would introduce some semantic anomalies. 
It
also ensures that evaluation is total---as we should expect. Given
that it is also deterministic, we can thus equivalently turn our
relational big-step semantics into a functional one: the unique trace
for a given statement and initial state is definable by corecursion.
(For details, see our previous paper \cite{NU:trabco}.)

The coinductive trace-based semantics agrees with the inductive
state-based semantics. 

\begin{prop}[\cite{NU:trabco}]
  For any $s$, $\st$, $\st'$, existence of $\tr$ such that
  $\exec{s}{\st}{\tr}$ and $\conv{\tr}{\st'}$ is equivalent to
  $\execind{s}{\st}{\st'}$.
\end{prop}

We notice that the inductive state-based semantics cannot be made
total constructively. It is unproblematic to complement the
inductively defined terminating evaluation relation with a
coinductively defined diverging evaluation relation, but this does not
help, as we cannot decide the halting problem.

\subsubsection*{Discussions on alternative designs}
We look at several seemingly
not so different but problematic alternatives that we reject,
thereby revealing some subtleties in designing
coinductive big-step semantics and motivating our design choices. 

Since progress of loops is not required for wellformedness of the
definitions of $\execnoargs$ and $\execseqnoargs$, one might be tempted 
to regard guard testing to be instantaneous and modify the rules for
the while-loop to take the form
\[
\infer={
  \exec{\While{e}{s_t}}{\st}{\tr'}
}{
  \istrue{e}{\st}
  &\exec{s_t}{\st}{\tr}
  &\execseq{\While{e}{s_t}}{\tr}{\tr'}
}
\quad
\infer={
  \exec{\While{e}{s_t}}{\st}{\sglt{\st}}
}{
  \isfalse{e}{\st}
}
\]
This leads to undesirable outcomes.  We can derive
$\exec{\While{\true}{\Skip}}{\st}{\sglt{\st}}$, which means that the
non-terminating $\While{\true}{\Skip}$ is considered semantically
equivalent to the terminal (immediately terminating) $\Skip$.  Worse,
we can also derive \linebreak $\exec{\Seq{\While{\true}{\Skip}}{x :=
    17}}{\st}{\cons{\st}{\sglt{\update{\st}{x}{17}}}}$, which is even
more inadequate: a sequence can continue to run after the
non-termination of the first statement. Yet worse, inspecting the
rules closer we discover we are also able to derive
$\exec{\While{\true}{\Skip}}{\st}{\tr}$ for any $\tr$.
Mathematically, giving up insisting on progress in terms of growing
the trace has also the consequence that the relational semantics
cannot be turned into a functional one, although While should
intuitively be total and deterministic. In a functional semantics,
evaluation must be a trace-valued function and in a constructive
setting such a function must be productive.

Another option, where assignments and test of guards are properly
taken to constitute steps, could be to define $\execseqnoargs$
by case distinction on the statement by rules such as
\[
\infer={
  \execseq{\While{e}{s_t}}{\tr}{\tr''}
}{
  \tr \models^* e
  &\execseq{s_t}{\duplast ~\tr}{\tr'}
  &\execseq{\While{e}{s_t}}{\tr'}{\tr''}
}
\quad
\infer={
  \execseq{\While{e}{s_t}}{\tr}{\duplast ~\tr}
}{
  \tr \not\models^* e
}
\]
Here, $\duplast ~\tr$, defined corecursively, traverses $\tr$ and
duplicates its last state, if it is finite.  Similarly, $\tr \models^*
e$ and $\tr \not\models^* e$ traverse $\tr$ and evaluate $e$ in the
last state, if it is finite:
\[
\infer={
\cons{\st}{\tr} \models^* e
}{
\tr \models^* e
}
\quad
\infer={
\sglt{\st} \models^* e
}{\st \models e}
\hspace*{1.5cm}
\infer={
\cons{\st}{\tr} \not\models^* e
}{
\tr \not\models^* e
}
\quad
\infer={
\sglt{\st} \not\models^* e
}{\st \not \models e}
\]
(The rules for $\Skip$ and sequence are very simple and appealing in
this design.) The relation $\execnoargs$ would then be defined
uniformly by the rule
\[
\infer{\exec{s}{\st}{\tr}}{
  \execseq{s}{\sglt{\st}}{\tr} 
}
\]
It turns out that we can still derive
$\exec{\While{\true}{\Skip}}{\st}{\tr}$ for any $\tr$. We can even derive $\exec{\While{\true}{x := x+1}}{\st}{\tr}$ for any $\tr$.

The third alternative (Leroy and Grall use this technique in
\cite{LG:coibso}) is most close to ours. It introduces, instead of
our $\execseqnoargs$ relation, an auxiliary relation 
$\mathit{split}$, defined coinductively by
\[
\infer={\split{\sglt{\st}}{\sglt{\st}}{\st}{\sglt{\st}}}{
}
\quad
\infer={\split{(\cons{\st}{\tr})}{\sglt{\st}}{\st}{(\cons{\st}{\tr'})}}{
  \bism{\tr}{\tr'}
}
\quad
\infer={\split{(\cons{\st}{\tr})}{(\cons{\st}{\tr_0})}{\st'}{\tr_1}}{
  \split{\tr}{\tr_0}{\st'}{\tr_1}
}
\]
so that $\split{\tr'}{\tr_0}{\st'}{\tr_1}$ expresses that the trace
$\tr'$ can be split into a concatenation of traces $\tr_0$ and $\tr_1$
glued together at a mid-state $\st'$. 
Then the evaluation relation is defined by replacing 
the uses of $\execseqnoargs$ with $\mathit{split}$, e.g., 
the rule for the sequence statement would be:
\[
\infer={
  \exec{\Seq{s_0}{s_1}}{\st}{\tr'}
}{
 \split{\tr'}{\tr_0}{\st'}{\tr_1}
  & \exec{s_0}{\st}{\tr_0}
  &\exec{s_1}{\st'}{\tr_1}
}
\]
This third alternative does not cause any outright anomalies for
While. But alarmingly $s_1$ has to be run from some (underdetermined)
state within a run of $s_0; s_1$ even if the run of $s_0$ does not
terminate. In a richer language with abnormal terminations, we get a
serious problem: no evaluation is derived for
$\Seq{(\While{\true}{\Skip})}{\Abort}$ although the $\Abort$ statement
should not be reached.


\section{Hoare logic}
\label{sec:hoare}

We now proceed to the Hoare logic and its soundness and completeness
proof. We base our consequence rule on semantic entailment rather than
derivability in some fixed proof system. This allows us to sidestep
the problem of its unavoidable incompleteness due to the impossibility
of complete axiomatization of any theory containing arithmetic. By
identifying assertions with state and trace predicates (more
precisely, predicates expressible in the meta-logic), we also avoid
the risk of possible incompleteness due to a chosen narrower assertion
language not being closed under weakest preconditions/strongest
postconditions. To formulate the rules of the Hoare logic, we
introduce a small set of assertion connectives, i.e., operations on
predicates. To be able to express the strongest postcondition of any
precondition, we need a few additional connectives.

\subsection{Assertions}

Our assertions are predicates over states and traces. A state
predicate $U$ is any predicate on states (in particular, it need not
be decidable). From a trace predicate $P$, we require additionally
that it is a setoid predicate, i.e., it must be unable to distinguish
bisimilar traces.

Although we refrain from introducing a language of assertions, we
introduce a number of connectives for our assertions, which are
operations on predicates.  All trace predicate connectives yield
setoid predicates. The inference rules of the Hoare logic make use of
these connectives. Indeed, it was an intriguing exercise for us to
come up with connectives that would be small but expressive enough for
practical specification purposes and at the same time allow us to
prove the Hoare logic sound and complete in our constructive setting.

The definitions of these connectives are given in
Figure~\ref{fig:assert}.\footnote{We use the symbol $\models$ to
  highlight application of a predicate to a state or trace. We are not
  defining a single satisfaction relation $\models$ for some assertion
  language, but a number of individual state/trace predicates and
  operations on such predicates.  Some of these operations are defined
  inductively, some coinductively, some definitions are not recursive
  at all.}

\begin{figure}
\[
\begin{array}{c}
\infer{
  \st \models \true
}{}
\qquad
\infer{
  \st \models \neg U
}{\neg (\st \models U)}
\qquad
\infer{
  \st \models U \wedge V
}{\st \models U & \st \models V}
\qquad
\ldots
\\[2ex]
\infer{
  \tr \models \true
}{}
\qquad
\infer{
  \tr \models \neg P
}{\neg (\tr \models P)}
\qquad
\infer{
  \tr \models P \wedge Q
}{\tr \models P & \tr \models Q}
\qquad
\ldots
\\[2ex]
\infer{
  \satisfy{\sglt{\st}}{\sglt{U}}
}{
  \satisfy{\st}{U}
}
\qquad
\infer{
  \satisfy{\tr'}{\Follows{P}{Q}}
}{
  \satisfy{\tr}{P}
  &\follows{Q}{\tr}{\tr'}
}
\qquad 
\infer={\satisfy{\tr}{\rep{P}}}{
  \satisfy{\tr}{\sglt{\true}}
}
\quad
\infer={\satisfy{\tr'}{\rep{P}}}{
  \satisfy{\tr}{P}
  &
  \follows{\rep{P}}{\tr}{\tr'}  
}
\\[2ex]
\infer{
  \satisfy{\cons{\st}{(\update{\st}{x}{\eval{e}{\st}})}}{\update{U}{x}{e}}
}{
  \satisfy{\st}{U}
}
\qquad
\infer{
  \satisfy{\cons{\st}{\sglt{\st}}}{\dup{U}}
}{
  \satisfy{\st}{U}
}
\\[2ex]
\infer{
  \satisfy{\st}{\Last{P}}
}{
  \satisfy{\tr}{P}
  &
  \conv{\tr}{\st}
}
\qquad
\infer{\satisfy{\tr}{\finite}}{
  \conv{\tr}{\st}
}
\qquad
\infer{\satisfy{\tr}{\infinite}}{
  \div{\tr} 
}
\\[2.5ex]
\infer=[\mathsf{[flw\mbox{-}nil]}]{
  \follows{Q}{\sglt{\st}}{\tr}
}{
  \hd~\tr = \st
  &\satisfy{\tr}{Q}
}
\quad
\infer=[\mathsf{[flw\mbox{-}delay]}]{
  \follows{Q}{\cons{\st}{\tr}}{\cons{\st}{\tr'}}
}{
  \follows{Q}{\tr}{\tr'}
}
\\[2.5ex]
\infer{U \models V}{
  \forall \st.\,  \satisfy{\st}{U} \to \satisfy{\st}{V}
}
\qquad
\infer{P \models Q}{
  \forall \tr.\, \satisfy{\tr}{P} \to \satisfy{\tr}{Q}
}

\end{array}
\]
\caption{Semantics of assertions}
\label{fig:assert}
\end{figure}

The two most primitive state (resp.\ trace) predicates are $\true$ and
$\false$, which are respectively true and false for any state (resp.\
trace). We can also use the standard connectives $\neg, \wedge, \vee$
and quantifiers $\forall, \exists$ to build state and trace
predicates. The context disambiguates the overloaded notations for
these state and trace predicates.

For a state predicate $U$, 
the singleton $\sglt{U}$ is a trace 
predicate that is true of singleton
traces given by a state satisfying $U$. 
In particular $\sglt{\true}$ is true of any singleton trace.

For a state predicate $U$, the doubleton  $\dup{U}$ is true of a
doubleton trace whose two states are identical and satisfy $U$. 

For a state predicate $U$, the update  $\update{U}{x}{e}$ 
is the strongest postcondition of
the statement $x := e$ for the precondition $U$. It is true of a
doubleton trace whose first state $\st$ satisfies $U$ and
second state is obtained from the first by modifying the value of $x$
to become $\sem{e}~\st$.

For trace predicates $P$ and $Q$, the chop $P \chop Q$ is a trace
predicate that is true, roughly speaking, of a trace $\tr'$ that has a
prefix $\tr$ satisfying $P$, with the rest of $\tr'$ satisfying $Q$.
(To be more precise, the prefix and the rest overlap on
a mid-state which is the last state of the prefix and the first state
of the suffix.)
But its definition is carefully crafted, so that $Q$ is not checked,
if $\tr$ is infinite (in which case necessarily $\bism{\tr}{\tr'})$,
and this happens without case distinction on whether $\tr$ is finite.
This effect is achieved with the premise $\follows{Q}{\tr}{\tr'}$. The
relation $\follows{Q}{\tr}{\tr'}$ is defined coinductively. It
traverses all of $\tr$, making sure that it is a prefix of $\tr'$ 
(rule $\mathsf{flw\mbox{-}delay}$), and, upon
possible exhaustion of $\tr$ in a finite number of steps, checks $Q$
against the rest of $\tr'$ (rule $\mathsf{flw\mbox{-}nil}$). 
This way the problem of deciding whether
$\tr$ is finite is avoided, basically by postponing it, possibly
infinitely.

Our chop operator is classically equivalent to the chop operator from
interval temporal logic \cite{Mos:temlrh,HZM:harsbt} (cf.\ also the
separating conjunction of separating logic). Indeed, classically,
$\satisfy {\tr'}{P \chop Q}$ holds iff
\begin{itemize}
\item either, for some finite prefix $\tr$ of $\tr'$, we have $\tr \models P$ and $\tr'' \models Q$, where $\tr''$ is the rest of $\tr'$,
\item or $\tr'$ is infinite and $\tr' \models P$.
\end{itemize}
This is how the semantics of chop is defined in interval temporal
logic. But it involves an upfront decision of whether $P$ will be
satisfied by a finite or an infinite prefix of $\tr'$. Our definition
is fine-tuned for constructive reasoning.


For a trace predicate $P$, its iteration $\rep{P}$ is a trace
predicate that is true of a trace which is a concatenation of a
possibly infinite sequence of traces, each of which satisfies
$P$. (This is modulo the overlap of the last and first states of
consecutive traces in the sequence and the empty concatenation being a
singleton trace.) It is reminiscent of the Kleene star operator.  It
is defined by coinduction and takes into account possibilities of both
infiniteness of some single iteration and infinite repetition.

For a trace predicate $P$, $\Last{P}$ is a state predicate
that is true of states that can be the last state of a finite trace
satisfying $P$. 

Trace predicates $\finite$ and $\infinite$ are true of 
finite and infinite traces, respectively. 

For state predicates $U$ and $V$ (resp. trace predicates $P$ and $Q$), 
$U \models V$ (resp. $P \models Q$) denotes entailment. 

\begin{prop}\label{prop:asserts_setoid}
  For any $U$, $\sglt{U}$, $\update{U}{x}{e}$ and $\dup{U}$ are setoid
  predicates. For any setoid predicates $P$, $Q$, $P \chop Q$ is a
  setoid predicate. For any setoid predicate $P$, $\rep{P}$ 
  is a setoid predicates. Moreover, $\finite$ and $\infinite$
  are setoid predicates. 
\end{prop}
\begin{proof}
That $\sglt{U}$, $\update{U}{x}{e}$ and $\dup{U}$ are setoid predicates
follows from the definition.  We prove that $P \chop Q$ and
$\rep{P}$ are setoid predicates when $P$ and $Q$ are by coinduction.
That $\finite$ and $\infinite$ are setoid predicates is proved
by induction and by coinduction respectively. 
\end{proof}

\begin{prop}\label{prop:asserts_monotone}
  For any $U$ and $V$, if $U \models V$,
  then $\sglt{U} \models \sglt{V}$, $\update{U}{x}{e} \models \update{V}{x}{e}$
  and $\dup{U} \models \dup{V}$. 
  For any setoid predicates $P, P'$ and $Q$, if $P \models P'$,
  then $P \chop Q \models P' \chop Q$ and
  $Q \chop P \models Q \chop P'$.
  For any setoid predicates $P$ and $Q$, if $P \models Q$,
  then $\rep{P} \models \rep{Q}$ and $\Last{P} \models \Last{Q}$.
\end{prop}
\begin{proof}
That $\sglt{U}$, $\update{U}{x}{e}$, $\dup{U}$ are monotone
follows from the definition.  We prove that $P \chop Q$ and
$\rep{P}$ are monotone and that $\Last{P}$ is monotone by
coinduction and by induction respectively. 
\end{proof}

A number of logical consequences and equivalences hold about these
connectives, to be proved in Lemma~\ref{lemma:misc_asserts}.
We have the trivial equivalence: 
$\sglt{\true} \chop P \logequ P
\logequ P \chop \sglt{\true}$.
The chop operator is associative: 
$(P \chop Q) \chop R \logequ P \chop
(Q \chop R)$. 
The iterator operator $\rep{P}$ repeats $P$ either zero times
or once followed by further repetitions: 
$\rep{P} \logequ \sglt{\true} \vee (P \chop \rep{P})$. 
A trace is infinite if and only if $\false$ holds
for any last state: $\infinite \logequ \true \chop \sglt{\false}$. 
If every trace satisfying $P$ is infinite, i.e., if $P
\models \infinite$, then any trace satisfying $P$ has no last state,
i.e., $\Last{P} \logequ \false$.
We have $P \chop \sglt{\Last{P}} \logequ P$, 
so that if a trace satisfies $P$, then its last state, 
if exists, satisfies $\Last{P}$.
The last state of a singleton trace $\sglt{\st}$ is $\st$, therefore
we have $\Last{\sglt{U}} \logequ U$.
We also have 
$\Last{(P \chop Q)} \models \Last{Q}$, but the converse does not hold. 
E.g., $\Last{\true} \not\models \Last{(\false \chop \true)}$.
Instead, we have
$\Last{(\sglt{\Last{P}} \chop Q)} \logequ \Last{(P \chop Q)}$.
Finally we have $\Last{(P \chop \sglt{U})} \logequ \Last{P} \wedge U$.
Namely, a state satisfies $U$ and can be the last state of a finite trace 
satisfying $P$ if and only if it can be the last state of 
a finite trace satisfying $P \chop \sglt{U}$.

We define the concatenation of traces $\tr$ and $\tr'$, 
$\tr \concat \tr'$, 
by replacing the last state of $\tr$ by $\tr'$. Formally, it is defined
by corecursion by
\[
\sglt{\st} \concat \tr = \tr
\qquad
(\cons{\st}{\tr}) \concat \tr' = \cons{\st}{(\tr \concat \tr')}
\]

We first observe three results, which are useful for later proofs.

\begin{lem}\label{lemma:follows_true}
For any $\tr$, $\follows{\sglt{\true}}{\tr}{\tr}$.
\end{lem}
\begin{proof}
By coinduction and case analysis on $\tr$. 
The case of $\bism{\tr}{\sglt{\st}}$ follows from 
$\sglt{\st} \models \sglt{\true}$.
The case of $\bism{\tr}{\cons{\st}{\tr'}}$: we get 
$\follows{\sglt{\true}}{\tr'}{\tr'}$ from the coinduction hypothesis,
from which we obtain $\follows{\sglt{\true}}{\tr}{\tr}$.
\end{proof}

\begin{lem}\label{lemma:follows_singleton}
For any $U$ and $\tr, \tr'$, if $\follows{\sglt{U}}{\tr}{\tr'}$,
then $\bism{\tr}{\tr'}$.
\end{lem}
\begin{proof}
By coinduction and inversion on $\follows{\sglt{U}}{\tr}{\tr'}$.
\end{proof}

\begin{lem}\label{lemma:conv_last}
For any $U$ and $\tr$, 
if for any $\st$, $\conv{\tr}{\st}$ implies $\satisfy{\st}{U}$,
then $\follows{\sglt{U}}{\tr}{\tr}$.
\end{lem}
\begin{proof}
By coinduction with case analysis on $\tr$.

The case of $\bism{\tau}{\sglt{\st}}$: We have $\conv{\tr}{\st}$,
hence $\st \models U$ by our hypothesis. 
We conclude $\follows{\sglt{U}}{\tr}{\tr}$.

The case of $\bism{\tr}{\cons{\st'}{\tr'}}$: Since 
$\conv{\tr'}{\st}$ implies
$\conv{\tr}{\st}$, we have that, 
for any $\st$, $\conv{\tr'}{\st}$ implies $\satisfy{\st}{U}$.
We get $\follows{\sglt{U}}{\tr'}{\tr'}$ by the coinduction hypothesis,
from where we conclude $\follows{\sglt{U}}{\tr}{\tr}$.
\end{proof}



\begin{lem}\label{lemma:misc_asserts}
For any $U, V$ and setoid predicates $P, Q$ and $R$, we have
\begin{enumerate}[label=\enspace(\arabic*)]
\item\label{lemma:sglt_dup}%

$\sglt{U} \chop \dup{V} \logequ \dup{U \wedge V}
\logequ \dup{U} \chop \sglt{V}$

\item\label{lemma:sglt_chop}%
$\sglt{U} \chop \sglt{V} \logequ \sglt{U \wedge V}$

\item\label{lemma:sglt_true_chop}%
$\sglt{\true} \chop P \logequ P
\logequ P \chop \sglt{\true}$

\item\label{lemma:chop_assoc}%

$(P \chop Q) \chop R \logequ P \chop
(Q \chop R)$

\item\label{lemma:rep_unfold}%

$\rep{P} \logequ \sglt{\true} \vee (P \chop \rep{P})$

\item\label{lemma:rep_idem}%
$\rep{P} \logequ \rep{P} \chop \rep{P}$

\item $\infinite \logequ \true \chop \sglt{\false}$

\item 
If $P \models \infinite$, then $\Last{P} \logequ \false$.

\item $P \logequ P \chop \sglt{\Last{P}}$
\end{enumerate}
\begin{enumerate}[label=\enspace(\arabic*),start=10]

\item\label{lemma:last_sglt}%
$\Last{\sglt{U}} \logequ U$

\item\label{lemma:last_chop}%
$\Last{(P \chop Q)} \models \Last{Q}$

\item\label{lemma:last_last}%
$\Last{(\sglt{\Last{P}} \chop Q)} \logequ \Last{(P \chop Q)}$

\item\label{lemma:last_chop_sglt}%
 $\Last{(P \chop \sglt{U})} \models  U$

\item\label{lemma:last_rep}
$\Last{(\sglt{I} \chop \rep{(P \chop \dup{I})})}
\models I$

\end{enumerate}
\end{lem}
\begin{proof}\hfill
\begin{enumerate}[label=\enspace(\arabic*)]
\item Follows from the definition.

\item Follows from the definition. 

\item 
$\sglt{\true} \chop P \logequ P$ follows from the definition.
$P \models P \chop \sglt{\true}$ holds by Lemma~\ref{lemma:follows_true}.
Suppose $\tr' \models P \chop \sglt{\true}$.
There exists $\tr$ such that $\tr \models P$
and $\follows{\sglt{\true}}{\tr}{\tr'}$.
By Lemma~\ref{lemma:follows_singleton} $\bism{\tr}{\tr'}$ holds,
so we must have $\tr' \models P$ since $P$ is a setoid predicate. 
This proves $P \chop \sglt{\true} \models P$.

\item 
Suppose $\tr'' \models (P \chop Q) \chop R$.
There exist $\tr$ and $\tr'$ such that
$\tr \models P$ and $\follows{Q}{\tr}{\tr'}$
and $\follows{R}{\tr'}{\tr''}$. 
We prove, 
for any $\tr_0, \tr_1$ and $\tr_2$, 
$\follows{Q}{\tr_0}{\tr_1}$ and $\follows{R}{\tr_1}{\tr_2}$
imply $\follows{Q \chop R}{\tr_0}{\tr_2}$ by coinduction
and inversion on $\follows{Q}{\tr_0}{\tr_1}$.
This yields $\follows{Q \chop R}{\tr}{\tr''}$ therefore
$\tr'' \models P \chop (Q \chop R)$.

The converse is more subtle. Given $\tr'' \models P \chop (Q \chop R)$,
we have to find a prefix $\tr'$ of $\tr''$ that satisfies $P \chop Q$
while $\follows{R}{\tr'}{\tr''}$.
To do so, we define a function 
$\midp: (\follows{P_0 \chop Q_0}{\tr_0}{\tr_1}) \rar \trace$ by corecursion
(we take $\tr_0$ and $\tr_1$ to be implicit parameters of $\midp$,
inferred from the proof argument).\footnote{Where useful, we give hypothetical proofs names, like $h$ below; $\_$ is for an anonymous dummy argument. $\mathit{existT}$ is the constructor of sigma-types in Coq.}
\[
\begin{array}{l}
\midp~(\mathsf{flw\mbox{-}nil} 
~\st~\tr_0~(\_:\hd ~\tr_0 = \st)~(h:\tr_0 \models P_0 \chop Q_0))\\
\quad = \mathit{let~existT} ~\tr_1 
~(\_: \tr_1 \models P_0 \wedge \follows{Q_0}{\tr_1}{\tr_0}) 
= h ~\mathit{in}~\tr_1\\[1ex]
\midp~(\mathsf{flw\mbox{-}delay} 
~\st~\tr_0~\tr_1~(h:\follows{P_0\chop Q_0}{\tr_0}{\tr_1}))
= \cons{\st}{\midp~h}
\end{array}
\]
We then prove that,
for any $\tr_0, \tr_1$ and $h:\follows{P_0\chop Q_0}{\tr_0}{\tr_1}$,
$\follows{P_0}{\tr_0}{\midp~h}$
and $\follows{Q_0}{\midp~h}{\tr_1}$ hold
by coinduction and inversion on $h$.

Now assume $\tr'' \models P \chop (Q \chop R)$.
There exists $\tr$ such that $\tr \models P$
and $h:\follows{Q\chop R}{\tr}{\tr''}$.
We have $\midp~h \models P \chop Q$, since 
$\follows{Q}{\tr}{\midp~h}$.
This together with $\follows{R}{\midp~h}{\tr''}$ proves
$\tr'' \models (P\chop Q) \chop R$, as required. 

\item Follows from the definition.

\item Suppose $\tr \models \rep{P}$. We have to prove
$\tr \models \rep{P} \chop \rep{P}$. From Lemma~\ref{lemma:follows_true}
and (5), we deduce $\follows{\rep{P}}{\tr}{\tr}$,
which gives us $\tr \models \rep{P} \chop \rep{P}$.
Conversely, suppose $\tr \models \rep{P} \chop \rep{P}$.
There exists $\tr'$ such that $\tr' \models \rep{P}$ and
$\follows{\rep{P}}{\tr'}{\tr}$. We close the case by proving the following
two conditions by mutual coinduction\footnote{In Coq, we actually perform
nested coinduction.}:
\begin{enumerate}
\item $\forall \tr\, \tr'.\, \tr \models \rep{P} \imp
\follows{\rep{P}}{\tr}{\tr'} \imp  \tr' \models \rep{P}$
\item $\forall \tr\, \tr'\, \tr''.\, 
\follows{\rep{P}}{\tr}{\tr'} \imp \follows{\rep{P}}{\tr'}{\tr''}
\imp \follows{\rep{P}}{\tr}{\tr''}$.
\end{enumerate}

(a): We perform inversion on $\tr \models \rep{P}$.
The case of $\tr \models \sglt{\true}$, i.e., $\bism{\tr}{\sglt{\st}}$:  
From $\follows{\rep{P}}{\tr}{\tr'}$, we conclude $\tr' \models \rep{P}$.
The case of $\tr'' \models P$ and $\follows{\rep{P}}{\tr''}{\tr}$:
we get $\follows{\rep{P}}{\tr''}{\tr'}$ by (b),
from where we conclude $\tr' \models \rep{P}$.

(b): We perform inversion on 
$\follows{\rep{P}}{\tr}{\tr'}$. The case of $\bism{\tr}{\sglt{\st}}$
and $\hd~\tr' = \st$ and $\tr' \models \rep{P}$ follows from
(a).
The case of $\bism{\tr}{\cons{\st}{\tr_0}}$ and
$\bism{\tr'}{\cons{\st}{\tr'_0}}$ and 
$\follows{\rep{P}}{\tr_0}{\tr'_0}$: 
We must have $\bism{\tr''}{\cons{\st}{\tr''_0}}$
and $\follows{\rep{P}}{\tr'_0}{\tr''_0}$.
The coinduction hypothesis (b) gives us
$\follows{\rep{P}}{\tr''_0}{\tr_0}$,
from which we conclude 
$\follows{\rep{P}}{\tr''}{\tr}$.

\item We prove an auxiliary condition: for any $\tr$,
$\infinite ~\tr$ iff $\follows{\sglt{\false}}{\tr}{\tr}$
by coinduction. 
$\infinite \models \true \chop \sglt{\false}$ follows
from the condition.
$\true \chop \sglt{\false} \models \infinite$ follows
from the condition and Lemma~\ref{lemma:follows_singleton}.

\item Follows from $\infinite \wedge \finite \models \false$.

\item 
Suppose $\tr \models P$.
By the definition of $\Last{P}$, we have 
for any $\st$, $\conv{\tr}{\st}$ implies $\satisfy{\st}{\Last{P}}$.
We then deduce $\follows{\sglt{\Last{P}}}{\tr}{\tr}$
by Lemma~\ref{lemma:conv_last}, thus
conclude $\satisfy{\tr}{P \chop \sglt{\Last{P}}}$. 

Conversely, suppose that $\tr' \models P \chop \sglt{\Last{P}}$, 
i.e., $\tr \models P$
and $\follows{\sglt{\Last{P}}}{\tr}{\tr'}$ for some $\tr$.
By Lemma~\ref{lemma:follows_singleton}, $\bism{\tr}{\tr'}$ holds,
so we must have $\tr' \models P$ since $P$ is a setoid predicate. 

\end{enumerate}
\begin{enumerate}[label=\enspace(\arabic*),start=10]
\item Follows from the definition.

\item Suppose $\st \models \Last{(P \chop Q)}$. 
There exist $\tr$ and $\tr'$ such that $\conv{\tr'}{\st}$
and $\tr \models P$ and $\follows{Q}{\tr}{\tr'}$.
We have to prove $\st \models \Last{Q}$. We do so by 
proving an auxiliary condition: for any $\st_0$ and $\tr_0$, 
if $\conv{\tr_0}{\st_0}$,
then for any $\tr_1$, $\follows{Q}{\tr_1}{\tr_0}$ implies
$\st_0 \models \Last{Q}$ by induction on the derivation of 
$\conv{\tr_0}{\st_0}$.

\item Suppose $\st \models \Last{(\sglt{\Last{P}} \chop Q)}$.
There exist $\tr$ and $\tr'$ 
such that $\tr \models P$, $\conv{\tr}{\hd~\tr'}$,
$\tr' \models Q$ and $\conv{\tr'}{\st}$.
We then have that the concatenation of $\tr$ and $\tr'$ has the
desired properties. Namely, 
$\tr \concat \tr' \models P \chop Q$
and $\conv{(\tr \concat \tr')}{\st}$. 
This proves $\st \models \Last{(P \chop Q)}$ as we wanted. 

Conversely, suppose $\st \models \Last{(P \chop Q)}$.
There exist $\tr$ and $\tr'$ such that $\tr \models P$,
$\follows{Q}{\tr}{\tr'}$ and $\conv{\tr'}{\st}$.
The finiteness of $\tr'$ implies that of $\tr$,
i.e., we have $\conv{\tr}{\st'}$ for some $\st'$.
We can therefore find the suffix $\tr''$ of $\tr'$ such that
$\bism{\tr'}{\tr \concat \tr''}$ and $\hd~\tr'' = \st'$. 
(Basically, we drop the first $n$ elements from $\tr'$ to obtain $\tr''$,
where $n$ is the length of $\tr$. Since $\tr$ is finite,
its length is defined.)
Together $\follows{Q}{\tr}{\tr'}$ and $\conv{\tr}{\st'}$
proves $\tr'' \models Q$. This concludes
$\st \models \Last{(\sglt{\Last{P}} \chop Q)}$,
as we wanted. 

%

\item By the monotonicity of the last and chop operators,
  it suffices to prove
  $\Last{(\true \chop \sglt{U})} \models U$. Suppose $\st \models
  \Last{(\true \chop \sglt{U})}$. There exists $\tr$ such that
  $\conv{\tr}{\st}$ and $\tr \models \true \chop \sglt{U}$. 
  We now prove $\forall \tr, \st.\,
  \conv{\tr}{\st} \to \forall \tr'.\, \follows{\sglt{U}}{\tr'}{\tr}
  \to \st \models U$ by induction on the proof of $\conv{\tr}{\st}$,
  from where $\st \models U$ follows. 
  
\item By (13) and the monotonicity
of the chop and iterator operators, it suffices to prove
$\sglt{U} \chop \rep{(\true \chop \dup{U})} \models \true \chop \sglt{U}$.
Suppose $\tr \models \sglt{U} \chop \rep{(\true \chop \dup{U})}$.
There exists $\tr'$ such that $\tr' \models \sglt{U}$
and $\follows{\rep{(\true \chop \dup{U})}}{\tr'}{\tr}$,
which give us $\tr \models \rep{(\true \chop \dup{U})}$
and $\hd\, \tr \models U$.
We want to prove $\follows{\sglt{U}}{\tr}{\tr}$.
We do so by proving the following conditions by mutual coinduction.
\begin{enumerate}
\item $\forall \tr.\, \hd\, \tr \models U \imp 
\tr \models \rep{(\true \chop \dup{U})} \imp
\follows{\sglt{U}}{\tr}{\tr}$

\item $\forall \tr, \tr', \tr''.\, 
\follows{\dup{U}}{\tr}{\tr'} \imp 
\follows{\rep{(\true \chop \dup{U})}}{\tr'}{\tr''}
\imp \follows{\sglt{U}}{\tr''}{\tr''}$.
\end{enumerate}

(a): We perform inversion on $\tr \models \rep{(\true \chop \dup{U})}$.
The case of $\tr \models \sglt{\true}$: We have $\bism{\tr}{\sglt{\st}}$.
This and $\hd\, \tr \models U$ prove
$\follows{\sglt{U}}{\tr}{\tr}$.
The case of $\tr' \models \true \chop \dup{U}$
and $\follows{\rep{(\true \chop \dup{U})}}{\tr'}{\tr}$:
The former gives us $\follows{\dup{U}}{\tr''}{\tr'}$
for some $\tr''$. We therefore conclude 
$\follows{\sglt{U}}{\tr}{\tr}$ by (b).

(b): We perform inversion on 
$\follows{\dup{U}}{\tr}{\tr'}$:
The case of $\bism{\tr}{\sglt{\st}}$
and $\hd~\tr' = \st$ and $\tr' \models \dup{U}$:
We have $\bism{\tr'}{\cons{\st}{\sglt{\st}}}$, 
therefore $\bism{\tr''}{\cons{\st}{\tr''_0}}$
for some $\tr''_0$ with $\hd~\tr''_0 = \st$, 
and $\tr''_0 \models \rep{(\true \chop \dup{U})}$.
From (a), we obtain $\follows{\sglt{U}}{\tr''_0}{\tr''_0}$, which yields $\follows{\sglt{U}}{\tr''}{\tr''}$, as required. 
The case of $\bism{\tr'}{\cons{\st}{\tr'_0}}$ and
$\bism{\tr}{\cons{\st}{\tr_0}}$ and
$\follows{\dup{U}}{\tr_0}{\tr'_0}$:
We have $\bism{\tr''}{\cons{\st}{\tr''_0}}$
and $\follows{\rep{(\true \chop \dup{U})}}{\tr'_0}{\tr''_0}$.
By (b), we get 
$\follows{\sglt{U}}{\tr''_0}{\tr''_0}$,
which yields $\follows{\sglt{U}}{\tr''}{\tr''}$.

\end{enumerate}
\end{proof}

\subsection{Inference rules}

The derivable judgements of the Hoare logic are given by the
inductively interpreted inference rules in Figure~\ref{fig:hoare}. The
proposition $\semax{U}{s}{P}$ states derivability of the judgement.
The intent is that $\semax{U}{s}{P}$ should be derivable precisely
when running a statement $s$ from an initial state satisfying $U$
is guaranteed to produce a trace satisfying $P$.

The rules for assignment and $\Skip$ are self-explanatory.

The rule for sequence is defined in terms of the chop operator.  The
precondition $V$ for the second statement $s_1$ is given by those
states in which a run of the first statement $s_0$ may terminate.
In particular, if $\semax{U}{s_0}{P}$
and $P \models \infinite$, i.e., 
$s_0$ is necessarily diverging for the precondition $U$,
then we have $\semax{U}{s_0}{P \chop \sglt{\false}}$.
In this case, from the derivability of 
$\semax{\false}{s_1}{Q}$ for any $Q$, we get $\semax{U}{s_0; s_1}{P \chop Q}$ for any $Q$. But this makes sense, since $P \chop Q \logequ P$ as soon as $P \models \infinite$.

The rule for if-statement uses the doubleton operator in accordance
with the operational semantics where we have chosen that testing the
boolean guard grows the trace.

The rule for while-statement is inspired by the corresponding rule of
the standard, state-based partial-correctness Hoare logic. It uses a
loop invariant $I$. This is a state predicate that has to be true each
time the boolean guard is about to be (re-)tested in a run of the
loop.  Accordingly, the precondition $U$ should be stronger then $I$.
Also, $I$ must hold each time an iteration of $s_t$ has finished, as
enforced by having $P \chop \sglt{I}$ as the postcondition of $s_t$.
The postcondition $\dup{U} \chop \rep{(\REMOVED{\sglt{e} \chop} P
  \chop \dup{I})} \chop \sglt{\REMOVED{I \wedge} \neg e}$ of the loop
consists of three parts.  $\dup{U}$ accounts for the first test of the
guard; $\rep{(\REMOVED{\sglt{e} \chop} P \chop \dup{I})}$ accounts for
iterations of the loop body in alternation with re-tests of the guard (notice that that we are again using the doubleton operator);
$\sglt{\REMOVED{I \wedge} \neg e}$ accounts for the state in which the
last test of the guard is finished.

We have chosen to introduce a separate rule for instantiating
auxiliary variables. Alternatively, we might have stated the
consequence rule in a more general form, as suggested by
Kleymann~\cite{Kle:auxvar}; yet the separation facilitates
formalization in Coq.

The various logical consequences and equivalences about the
connectives suggest also further alternative and equivalent
formulations.
For instance, we could replace the rule for the while-statement by
\[
  \infer{
    \semax{I}{\While{e}{s_t}}{
    \dup{I} \chop \rep{(P \chop \dup{I})}
                      \chop \sglt{\neg e}}}{
    \semax{\conj{e}{I}}{s_t}{\Follows{P}{\sglt{I}}}
  }
\]
if we strengthened the consequence rule to  
\[
  \infer{
    \semax{U}{s}{P}
  }{
   U \models U'
   &\semax{U'}{s}{P'}
   &
   \sglt{U} \chop P' \models P
  }
\]
With our chosen rule for while, this strengthened version of
consequence is admissible:
\begin{lem}\label{lemma:semax_hd}
For any $U$, $s$ and $P$,
if $\semax{U}{s}{P}$, then $\semax{U}{s}{\sglt{U} \chop P}$.
\end{lem}
\begin{proof}
We prove the following more general statement by
induction on the derivation of $\semax{U}{s}{P}$:
for any $U$, $s$ and $P$,
if $\semax{U}{s}{P}$, then for any $V$,
$\semax{\conj{U}{V}}{s}{\sglt{V} \chop P}$.
\end{proof}
We do not attempt to argue that our formulation is the best choice;
yet we found that the present formulation is viable from the
points-of-view of both the meta-theory and applicability of the logic.

\begin{figure}
\[
\begin{array}{c}
\infer{
  \semax{U}{\Assign{x}{e}}{\update{U}{x}{e}}
}{}
\quad
\infer{
  \semax{U}{\Skip}{\sglt{U}}
}{}
\quad
\infer{
  \semax{U}{\Seq{s_0}{s_1}}{\Follows{P}{Q}}
}{
    \semax{U}{s_0}{\Follows{P}{\sglt{V}}}
    &\semax{V}{s_1}{Q}
}
\\[2ex]
\infer{
  \semax{U}{\Ifthenelse{e}{s_t}{s_f}}{\Follows{\dup{U}}{P}}}{
  \semax{\conj{e}{U}}{s_t}{P}
  &
  \semax{\conj{\neg e}{U}}{s_f}{P}
}
\\[2ex]
\infer{
  \semax{U}{\While{e}{s_t}}{
  \dup{U} \chop \rep{(\REMOVED{\sglt{e} \chop} P \chop \dup{I})}
                    \chop \sglt{\REMOVED{I \wedge} \neg e}}}{
  U \models I
  &
  \semax{\conj{e}{I}}{s_t}{\Follows{P}{\sglt{I}}}
}
\\[2ex]
\infer{
  \semax{U}{s}{P}
}{
 U \models U'
 &\semax{U'}{s}{P'}
 &
 P' \models P
}
\quad
\infer{
  \semax{\exists z.\, U}{s}{\exists z.\, P}
}{
  \forall z.\, \semax{U}{s}{P}
}
\end{array}
\]
\caption{Inference rules of Hoare logic}
\label{fig:hoare}
\end{figure}

\subsection{Soundness}

The soundness result states that any derivable Hoare triple is
semantically valid in the sense that, if the precondition holds of the
initial state of an evaluation, then the postcondition is true of the
trace produced.

\begin{prop}[Soundness]\label{prop:sound}
  For any $s$, $U$ and $P$, if $\semax{U}{s}{P}$, then, for all $\st$ and
  $\tr$, $\satisfy{\st}{U}$ and $\exec{s}{\st}{\tr}$ imply
  $\satisfy{\tr}{P}$.
\end{prop}

\begin{proof}
By induction on the derivation of $\semax{U}{s}{P}$. 
We show the main cases of sequence and while. 
\begin{itemize}
\item $s = \Seq{s_0}{s_1}$: We are given as the induction hypothesis that,
for any $\st, \tr$, $\exec{s_0}{\st}{\tr}$ and $\satisfy{\st}{U}$
imply $\satisfy{\tr}{P \chop \sglt{V}}$, 
and that,
for any $\st, \tr$, $\exec{s_1}{\st}{\tr}$ and $\satisfy{\st}{V}$
imply $\satisfy{\tr}{Q}$. We have to prove $\satisfy{\tr_1}{P \chop Q}$,
given $\satisfy{\st}{U}$ and $\exec{s_0}{\st}{\tr_0}$
and $\execseq{s_1}{\tr_0}{\tr_1}$. 
The induction hypothesis for $s_0$ gives us 
$\satisfy{\tr_0}{P \chop  \sglt{V}}$.
By Lemma~\ref{lemma:follows_singleton} and that $P$ is 
a setoid predicate, 
we derive $h_0:\satisfy{\tr_0}{P}$ and $\follows{\sglt{V}}{\tr_0}{\tr_0}$.
We prove by coinduction an auxiliary lemma: for any $\tr$, $\tr'$, 
$\follows{\sglt{V}}{\tr}{\tr}$ 
and $\execseq{s_1}{\tr}{\tr'}$ give $\follows{Q}{\tr}{\tr'}$, 
using the induction hypothesis for $s_1$. 
The lemma gives us
$h_1:\follows{Q}{\tr_0}{\tr_1}$. 
We can now close the case by $h_0$ and $h_1$.

\item $s = \While{e}{s_t}$: 
We are given as the induction hypothesis
that for any $\st$ and $\tr$, $\satisfy{\st}{e \wedge I}$
and $\exec{\st}{s_t}{\tr}$ imply
$\satisfy{\tr}{P \chop \sglt{I}}$.
We also have $U \models I$ .
We have to prove $\satisfy{\tr}{\dup{U} 
\chop \rep{(P \chop \dup{I})} \chop \sglt{\neg e}}$,
given $\satisfy{\st}{U}$
and $\exec{\While{e}{s_t}}{\st}{\tr}$.
We prove that, for any $\st$ and $\tr$,
$\exec{\While{e}{s_t}}{\st}{\tr}$ implies
$\follows{\sglt{\neg e}}{\tr}{\tr}$ by coinduction. 
It remains to prove $\tr \models \dup{U} \chop \rep{(P \chop \dup{I})}$.
By inversion on $\exec{\While{e}{s_t}}{\st}{\tr}$, 
we learn that $\bism{\tr}{\cons{\st'}{\tr'}}$ for some
$\st'$ and $\tr'$ such that $\hd~\tr' = \st'$.
So, we close the case by 
proving the following conditions by mutual coinduction:
\begin{itemize}
\item for any $\st$ and $\tr$,
if $\satisfy{\st}{I}$ and $\exec{\While{e}{s_t}}{\st}{\cons{\st}{\tr}}$,
then $\satisfy{\tr}{\rep{(P \chop \dup{I})}}$

\item for any $\tr$ and $\tr'$,
if 
$\follows{\sglt{I}}{\tr}{\tr}$
and $\execseq{\While{e}{s_t}}{\tr}{\tr'}$,
then 
$\follows{\rep{\dup{I} \chop (P \chop \dup{I})}}{\tr}{\tr'}$.
\end{itemize}

\end{itemize}
\end{proof}

Thanks to Proposition~\ref{prop:total} (totality of evaluation), as an
immediate corollary of Proposition~\ref{prop:sound} (soundness) we
learn that, if a state satisfies the precondition of a derivable Hoare
triple, then there exists an evaluation producing a trace satisfying
the postcondition.

\begin{cor}[Total-correctness soundness] \label{coro:semax_satisfiable}
For any $s$, $U$ and $P$, if $\semax{U}{s}{P}$,
then, for any $\st$ such that $\st \models U$,
there exists $\tr$ such that $\exec{s}{\st}{\tr}$ and $\tr \models P$.
\end{cor}

\subsection{Completeness}

The completeness result states that any semantically valid Hoare
triple is derivable. Following the standard approach (see, e.g., ~\cite{NN:sem})
we define, for a given statement $s$ and a given precondition $U$, a trace predicate $\sp{s}{U}$---the candidate strongest postcondition.
Then we prove that $\sp{s}{U}$ is a postcondition according to the
logic (i.e., $\semax{U}{s}{\sp{s}{U}}$ is derivable) and that
$\sp{s}{U}$ is semantically stronger than any other trace predicate that is a postcondition semantically. Completeness follows.

The trace predicate $\sp{s}{U}$ is defined by recursion on $s$ in
Figure~\ref{fig:sp}. The definition is mostly self-explanatory, as it
mimics the inference rules of the logic, except that we need the
loop-invariant $\Inv{e}{s}{U}$.  $\Inv{e}{s}{U}$ characterizes the set
of states that are reachable by some run of \While{e}{s_t} from a
state satisfying $U$ and where the boolean guard is tested in that
run.\footnote{Because of the induction-recursion involved in the
  simultaneous definition of $\Inv{e}{s}{U}$ and $\sp{s}{U}$, we have
  used impredicativity in our Coq development.}

\begin{figure}
\[
\begin{array}{l}
\sp{\Assign{x}{e}}{U} =  \update{U}{x}{e}\\
\sp{\Skip }{U} =  \sglt{U}\\
\sp{\Seq{s_0}{s_1}}{U} = 
   P \chop \sp{s_1}{\Last{P}} \mathit{~where~} P = \sp{s_0}{U}  \\
\sp{\Ifthenelse{e}{s_t}{s_f}}{U}  =
  \dup{U} \chop (\disj{\sp{s_t}{\conj{e}{U}}}{\sp{s_f}{\conj{\neg e}{U}}}) \\
\sp{\While{e}{s_t}}{U} = 
  \dup{U} \chop \rep{(\sp{s_t}{\conj{e}{I}} \chop \dup{I})} \chop \sglt{\neg e}\\
\hspace*{2cm} \mathit{~where~} I = \Inv{e}{s_t}{U} \\
\\[2ex]
\infer{ \satisfy{\st}{\Inv{e}{s}{U}} }{
  \satisfy{\st}{U}
}
\quad
\infer{ \satisfy{\st}{\Inv{e}{s}{U}} }{
  V \models \Inv{e}{s}{U}
  &
  \satisfy{\st}{\Last{(\sglt{\Inv{e}{s}{U} \wedge e} \chop \sp{s}{V})}}
}
\end{array}
\]
\caption{Strongest postcondition}\label{fig:sp}
\end{figure}

For any $s$ and $U$, the predicate $\sp{s}{U}$ is a monotone setoid predicate.

\begin{lem}\label{lemma:sp_setoid}
For any $s$, $U$, $\tr$, $\tr'$,
if $\satisfy{\tr}{\sp{s}{U}}$ and $\bism{\tr}{\tr'}$, then 
$\satisfy{\tr'}{\sp{s}{U}}$.
\end{lem}
\begin{proof}
By induction on the structure of $s$. 
\end{proof}

\begin{lem}\label{lemma:sp_cont}
For any $s$, $U$, $U'$,
if $U \models U'$, then $\sp{s}{U} \models \sp{s}{U'}$.
\end{lem}
\begin{proof}
By induction on the structure of $s$.
\end{proof}

The following lemma states that 
any trace which satisfies $\sp{s}{U}$ has its first state
satisfying $U$. 

\begin{lem}\label{lemma:sp_hd}
For any $s$, $U$, $\tr$, if $\satisfy{\tr}{\sp{s}{U}}$, then $\satisfy{\hd~\tr}{U}$.
\end{lem}
\begin{proof}
By induction on the structure of $s$.
\end{proof}

The next lemma states a crucial property of $\Inv{e}{s}{U}$.

\begin{lem}\label{lemma:Inv_correct}
For any $s$, $e$, $U$, 
$\sp{s}{\conj{\Inv{e}{s}{U}}{e}} \logequ
\sp{s}{\Inv{e}{s}{U} \wedge e} \chop \sglt{\Inv{e}{s}{U}}$.
\end{lem}

\begin{proof}
($\Rightarrow$): Suppose $\tr \models \sp{s}{\conj{\Inv{e}{s}{U}}{e}}$.
It suffices to prove $\follows{\sglt{\Inv{e}{s}{U}}}{\tr}{\tr}$.
We have $\hd\, \tr \models \conj{\Inv{e}{s}{U}}{e}$
by Lemma~\ref{lemma:sp_hd}, and
$\tr \models \sp{s}{\Inv{e}{s}{U}}$ by Lemma~\ref{lemma:sp_cont}
and $\conj{\Inv{e}{s}{U}}{e} \models \Inv{e}{s}{U}$.
These give us $\tr \models \sglt{\Inv{e}{s}{U} \wedge e} \chop
\sp{s}{\Inv{e}{s}{U}}$. 
By the definition of {\it Inv}, we have for
any $\st$, $\conv{\tr}{\st}$ implies $\satisfy{\st}{\Inv{e}{s}{U}}$.
Therefore we conclude $\follows{\sglt{\Inv{e}{s}{U}}}{\tr}{\tr}$ by
Lemma~\ref{lemma:conv_last}.  

($\Leftarrow$): Suppose
$\tr \models \sp{s}{\Inv{e}{s}{U} \wedge e} \chop
\sglt{\Inv{e}{s}{U}}$. We then have some $\tr'$ such that
$\tr' \models \sp{s}{\Inv{e}{s}{U} \wedge e}$
and $\follows{\sglt{\Inv{e}{s}{U}}}{\tr'}{\tr}$.
The latter proves $\bism{\tr}{\tr'}$
by Lemma~\ref{lemma:follows_singleton}.
We conclude $\tr \models \sp{s}{\conj{\Inv{e}{s}{U}}{e}}$
by Lemma~\ref{lemma:sp_setoid}.
\end{proof}


We are now ready to establish that $\sp{s}{U}$ is a postcondition
according to the Hoare logic.

\begin{lem}\label{lemma:a}
For any $s$, $U$, we have $\semax{U}{s}{\sp{s}{U}}$.
\end{lem}

\begin{proof}
By induction on $s$. We show the main cases of sequence and while.
\begin{itemize}
\item $s = \Seq{s_0}{s_1}$:
We are given as the induction hypotheses that,
for any $U_0$,  \linebreak $\semax{U_0}{s_0}{\sp{s_0}{U_0}}$ and 
$\semax{U_0}{s_1}{\sp{s_1}{U_0}}$.
We have to prove $\semax{U}{\Seq{s_0}{s_1}}{P \chop \sp{s_1}{\Last{P}}}$
where $P = \sp{s_0}{U}$.
By the induction hypothesis, we have $\semax{U}{s_0}{P}$,
thus $\semax{U}{s_0}{P\chop \sglt{\Last{P}}}$ by
(9) of Lemma~\ref{lemma:misc_asserts} and 
the consequence rule. We therefore 
close the case with $\semax{\Last{P}}{s_1}{\sp{s_1}{\Last{P}}}$
given by the induction hypothesis. 

\item $s = \While{e}{s_t}$:
We are given as the induction hypothesis that
$\semax{U_0}{s_t}{\sp{s_t}{U_0}}$, for any $U_0$. 
We have to prove 
$\semax{U}{\While{e}{s_t}}
{\dup{U} \chop \rep{(\sp{s_t}{\conj{e}{I}} \chop \dup{I})} 
\chop \sglt{\REMOVED{I \wedge} \neg e}}$
where $I = \Inv{e}{s_t}{U}$.
It is sufficient to prove
$\semax{\conj{e}{I}}{s_t}{(\sp{s_t}{\conj{e}{I}} \chop \sglt{I})}$,
which follows from the induction hypothesis and Lemma~\ref{lemma:Inv_correct}.
\qedhere
\end{itemize}
\end{proof}

\noindent Following the standard route, it remains to prove the following condition: 
for any $s$, $U$, $P$, if 
for all $\st, \tr$, $\satisfy{\st}{U}$ and $\exec{s}{\st}{\tr}$ imply 
$\satisfy{\tr}{P}$, then $\sp{s}{U} \models P$.

This will be an immediate corollary from Lemma~\ref{lemma:sp_hd}
and the following lemma, stating that any trace satisfying $\sp{s}{U}$ 
is in fact produced by a run of $s$.

\begin{lem}\label{lemma:b}
For any $s$, $U$, $\tr$, if $\satisfy{\tr}{\sp{s}{U}}$
then $\exec{s}{\mathit{hd} ~\tr}{\tr}$.
\end{lem}

\begin{proof}
By induction on $s$. We show the main cases of sequence and while.
\begin{itemize}

\item $s = \Seq{s_0}{s_1}$: We are given as  the
induction hypotheses that, for any $U', \tr'$,
$\satisfy{\tr'}{\sp{s_0}{U'}}$ (resp. $\satisfy{\tr'}{\sp{s_1}{U'}}$)
implies 
$\exec{s_0}{\mathit{hd}~\tr'}{\tr'}$
(resp. $\exec{s_1}{\mathit{hd}~\tr'}{\tr'}$).
We have to prove $\exec{\Seq{s_0}{s_1}}{\mathit{hd}~\tr}{\tr}$,
given $\satisfy{\tr}{\sp{\Seq{s_0}{s_1}}{U}}$, which
unfolds into $\satisfy{\tr_0}{\sp{s_0}{U}}$
and $\follows{\sp{s_1}{\Last{(\sp{s_0}{U})}}}{\tr_0}{\tr}$.
By the induction hypothesis for $s_0$, we have $\exec{s_0}{\mathit{hd}~\tr_0}{\tr_0}$.
Using the induction hypothesis for $s_1$, we prove by coinduction that,
for any  $\tr_1, \tr_2$, 
$\follows{\sp{s_1}{\Last{(\sp{s_0}{U})}}}{\tr_1}{\tr_2}$
implies $\execseq{s_1}{\tr_1}{\tr_2}$, thereby we close the case. 

\item $s = \While{e}{s_t}$:
We are given as the induction hypothesis that,
for any $U', \tr'$, $\satisfy{\tr'}{\sp{s_t}{U'}}$
implies $\exec{s_t}{\mathit{hd}~\tr'}{\tr'}$.
We have to prove
$\exec{\While{e}{s_t}}{\mathit{hd}~\tr}{\tr}$, given
$\satisfy{\tr} 
{\dup{U} \chop \rep{(\sp{s_t}{\conj{e}{I}} \chop \dup{I})} \chop  \sglt{\neg e}}$
where $I = \Inv{e}{s_t}{U}$.
We do so by proving the following two conditions simultaneously 
by mutual coinduction:
\begin{itemize}

\item for any $\tr$,
$\satisfy{\tr}
{\rep{(\sp{s_t}{\conj{e}{I}} \chop \dup{I})} \chop  \sglt{\neg e}}$
implies
$\exec{\While{e}{s_t}}{\hd~\tr}{\cons{\hd~\tr}{\tr}}$,


\item for any $\tr$ and $\tr'$, 
$\follows{\dup{I} \chop \rep{(\sp{s_t}{\conj{e}{I}} \chop \dup{I})} 
\chop  \sglt{\neg e}}{\tr}{\tr'}$
implies \linebreak $\execseq{\While{e}{s_t}}{\tr}{\tr'}$.\qedhere
\end{itemize}
\end{itemize}
\end{proof}

\begin{cor}\label{cor}
For any $s$, $U$ and $P$, if 
for all $\st$ and $\tr$, $\satisfy{\st}{U}$ and $\exec{s}{\st}{\tr}$ imply 
$\satisfy{\tr}{P}$, then $\sp{s}{U} \models P$.
\end{cor}

Completeness is proved as a corollary of 
Lemma~\ref{lemma:a} and Corollary~\ref{cor}.

\begin{prop}[Completeness]\label{prop:complete}
For any $s$, $U$ and $P$, if 
for all $\st$ and $\tr$, $\satisfy{\st}{U}$ and $\exec{s}{\st}{\tr}$ imply 
$\satisfy{\tr}{P}$, then $\semax{U}{s}{P}$.
\end{prop}

\begin{proof}
  Assume that for all $\st, \tr$, $\satisfy{\st}{U}$ and
  $\exec{s}{\st}{\tr}$ imply $\satisfy{\tr}{P}$. 
  By Corollary~\ref{cor}, we have that $\sp{s}{U} \models P$.
  By Lemma~\ref{lemma:a}, we have $\semax{U}{s}{\sp{s}{U}}$. Applying
  consequence, we get $\semax{U}{s}{P}$.
\end{proof}

Combining Propositions~\ref{prop:determ} (determinacy of evaluation)
and \ref{prop:complete} (completeness), we immediately get
completeness for total correctness.

\begin{cor}[Total-correctness completeness]
For any $s$, $U$ and $P$, if, for all $\st$ such that
$\satisfy{\st}{U}$, there is $\tr$ such that $\exec{s}{\st}{\tr}$ and
$\satisfy{\tr}{P}$, then $\semax{U}{s}{P}$.
\end{cor}

\section{Relation to the standard partial-correctness and
  total-correctness Hoare logics}\label{sec:Hoare}

It is easy to see, by going through the soundness and completeness
results, that our trace-based Hoare logic is a conservative extension
of the standard, state-based partial-correctness and total-correctness
Hoare logics. But more can be said. The derivations in these two
logics are directly transformable into derivations in our logic,
preserving their structure, without invention of new invariants or
variants. And in the converse direction, derivations in our logic are
transformable into derivations into the standard logics in a way that
removes from postconditions information about intermediate states. In
this direction, the variant for a while-loop is obtained by bounding
the length of traces satisfying the trace invariant of the loop. 

Concerning total correctness, we use two variations of the while-rule.
In the forward transformation, we use a version of the while-rule with
a dedicated variant (a natural-valued \emph{function} on states)
whereas, in the backward transformation, we work with a version where
the invariant (a state predicate) is made dependent on a natural
number (i.e., becomes a \emph{relation} between states and naturals;
crucially, there is no functionality requirement: in the same state,
the invariant can be satisfied by zero or one or several naturals). 
The two alternative
while-rules for total correctness are:
\[
\infer[\whilefun]{
  \semax{I \wedge t = m}{\While{e}{s_t}}{I \wedge t \leq m \wedge \neg e}
}{
  \forall n:\mathit{nat}.
  ~\semax{e \wedge I \wedge t = n}{s_t}{I \wedge t < n}
}
\]
and
\[
\infer[\whilerel]{
  \semax{J\, m}{\While{e}{s_t}}
        {\exists k.\, k \leq m \wedge J\, k \wedge \neg e}
}{
  \forall n:\mathit{nat}.
  ~\semax{e \wedge J\, n}{s_t}{\exists k.\, k < n \wedge J\, k}
}
\]
There is a reason for this discrepancy, which reflects our compromise
between pursuing a constructive approach and striving for purely syntactic
translations. We will discuss it in Section~\ref{sec:Hoare:discussion}
after having presented the transformations.

We have tried to fine-tune the inference rules in the different Hoare
logics and the transformations between them for smoothness. There is
some room for variations in them. The transformations are quite
sensitive to the exact division of labor in the source and target
Hoare logics between the rules for the statement constructors and the
consequence rule, but the effects of the possible variations are
mostly inessential.

For reference, the inference rules of the state-based logics appear in
the Appendix. Notice that also here we use predicates as assertions
and entailment as consequence, so there is no dedicated assertion
language or proof system for assertions.

\subsection{Embeddings of the standard Hoare logics into the
  trace-based logic}

We formalize our claim of embeddability of the standard Hoare logics
in the following two propositions, whose direct proofs are algorithms
for the transformations.

Proposition~\ref{prop:partial_Hoare} states that, if $\semax{U}{s}{Z}$
is a derivable partial-correctness judgement, then $\semax{U}{s}{\true
  \chop \sglt{Z}}$ is derivable in our logic.  The trace predicate
$\true \chop \sglt{Z}$ indicates that $Z$ holds of any state that is
reachable by traversing, in a finite number of steps, the whole trace
$\tr$ produced by running $s$.  Classically, this amounts to 
$Z$ being true of the last state of $\tr$, if $\tr$ is
finite and hence has one; if $\tr$ is infinite, then nothing is
required.  

Proposition~\ref{lemma:total_Hoare} states that, if $\semax{U}{s}{Z}$
is a derivable total-correctness judgement, then $\semax{U}{s}{\finite
  \chop \sglt{Z}}$ is derivable in our logic (in fact, it states a
little more).  The trace predicate $\finite \chop \sglt{Z}$ expresses
that the trace $\tr$ produced by running $s$ is finite and $Z$ holds
of the last state of $\tr$; the finiteness of $\tr$ guarantees the
existence of this last state.

\begin{prop}\label{prop:partial_Hoare}
  For any $s$, $U$ and $Z$, if $\semax{U}{s}{Z}$ is derivable in the
  partial-correctness Hoare logic, then $\semax{U}{s}{\true \chop
    \sglt{Z}}$ is derivable in the trace-based Hoare logic.
\end{prop}
\begin{proof}
  By induction on the derivation of $\semax{U}{s}{Z}$.  We show the
  main cases of sequence and while.
\begin{itemize}
\item $s = \Seq{s_0}{s_1}$:
We are given as the induction hypotheses 
$\semax{U}{s_0}{\true \chop \sglt{V}}$
and $\semax{V}{s_1}{\true \chop \sglt{Z}}$.
We have to prove $\semax{U}{\Seq{s_0}{s_1}}{\true \chop \sglt{Z}}$,
which is derived by
\[
\infer[(1)]{ \semax{U}{\Seq{s_0}{s_1}}{\true \chop \sglt{Z}}}{
  \infer{ \semax{U}{\Seq{s_0}{s_1}}{\true \chop  \true \chop \sglt{Z}}}{
    \infer*[\IH_0]{\semax{U}{s_0}{\true \chop \sglt{V}}}{
    }
    &
    \infer*[\IH_1]{\semax{V}{s_1}{\true \chop \sglt{Z}}}{
    }
  }
}
\]
\begin{enumerate}
\item We have
\[
\true \chop \true 
\logequ 
\true
\]
\end{enumerate}

\item $s = \While{e}{s_t}$:
We are given as the induction hypothesis 
$\semax{e \wedge I}{s_t}{\true \chop \sglt{I}}$.
We have to prove 
$\semax{I}{\While{e}{s_t}}{\true \chop \sglt{I \wedge \neg e}}$,
which is derived by
\[
\infer[(1)]{
  \semax{I}{\While{e}{s_t}}{\true \chop \sglt{I \wedge \neg e}}
}{
  \infer{ 
    \semax{I}{\While{e}{s_t}}
    {\dup{I} \chop \rep{(\REMOVED{e \chop} \true \chop \dup{I})} \chop \sglt{\REMOVED{I \wedge} \neg e}}
  }{
    \infer*[\IH_t]{\semax{e \wedge I}{s_t}{\true \chop \sglt{I}}}{
    }
  }
}
\]
\begin{enumerate}
\item We have 
\begin{eqnarray*}
\lefteqn{\dup{I} \chop \rep{(\true \chop \dup{I})} \chop \sglt{\neg e}} \\
& \entails &
\true \chop \sglt{I} \chop \sglt{\neg e}\\
& \logequ &
\true \chop \sglt{I \wedge \neg e} 
\qquad \textrm{(by Lemma~\ref{lemma:misc_asserts} (2))}
\end{eqnarray*}
\end{enumerate}
\end{itemize}\vspace{-12 pt}
\end{proof}

\noindent For the embedding of total-correctness derivations, we prove a
slightly stronger statement to have the induction go through.

\begin{prop}\label{lemma:total_Hoare}
  For any $s$, $U$ and $Z$, if $\semax{U}{s}{Z}$ is derivable in the
  total-correctness Hoare logic with $\whilefun$, then for any $W$,
  $\semax{U \wedge W}{s}{\sglt{W} \chop \finite \chop \sglt{Z}}$ is
  derivable in the trace-based Hoare logic.
\end{prop}
\begin{proof}
By induction on the derivation of $\semax{U}{s}{Z}$.
We show the main cases of sequence and  while.
\begin{itemize}
\item $s = \Seq{s_0}{s_1}$:
We are given as 
the induction hypotheses that,
for any $W_0$, 
$\semax{U \wedge W_0}{s_0}{\sglt{W_0}\chop \finite \chop \sglt{V}}$
and, for any $W_1$,
$\semax{V \wedge W_1}{s_1}{\sglt{W_1}\chop \finite \chop \sglt{Z}}$.
We have to prove that, for any $W$,
$\semax{U \wedge W}{\Seq{s_0}{s_1}}
{\sglt{W} \chop \finite \chop \sglt{Z}}$. This is done by the
derivation
\[
\infer[(1)]{ \semax{U \wedge W}{\Seq{s_0}{s_1}}
              {\sglt{W} \chop \finite \chop \sglt{Z}}}{
  \infer{\semax{U \wedge W}{\Seq{s_0}{s_1}}
         {\sglt{W} \chop \finite \chop \sglt{V} \chop \finite \chop \sglt{Z}}}{
    \infer*[\IH_0\, W]{\semax{U \wedge W}{s_0}{\sglt{W} \chop \finite \chop \sglt{V}}}{
    }
    &
    \infer{\semax{V}{s_1}{\sglt{V} \chop \finite \chop \sglt{Z}}}{
      \infer*[\IH_1\, V]{\semax{V \wedge  V}{s_1}{\sglt{V} \chop \finite \chop \sglt{Z}}}{
      }
    }
  }
}
\]
\begin{enumerate}
\item We have 
\[
\finite \chop \sglt{V} \chop \finite
\entails
\finite \chop \finite
\logequ
\finite
\]
\end{enumerate}

\item $s = \While{e}{s_t}$:
We are given as the induction hypothesis that, for all $n : nat$ and $W_t$, 
$\semax{e \wedge I \wedge t = n \wedge W_t}{s_t}
{\sglt{W_t} \chop \finite \chop \sglt{I \wedge t < n}}$.
We have to prove that, for any $W$, 
$\semax{I \wedge t = m \wedge W}{\While{e}{s_t}}
{\sglt{W} \chop \finite \chop \sglt{I \wedge \neg e \wedge t \leq m}}$. This is accomplished
by the derivation 
\[
\infer[(1)]{
  \semax{I \wedge t = m \wedge W}{\While{e}{s_t}}
        {\sglt{W} \chop \finite \chop \sglt{I \wedge \neg e \wedge t \leq m}}
}{
  \infer{
    \twoline{\{I \wedge t = m \wedge W\} ~\While{e}{s_t}}
     {\{\dup{I \wedge t = m \wedge W} \chop 
           \rep{((\exists n. \sglt{t = n} \chop \finite \chop 
           \sglt{t < n}) \chop \dup{I \wedge t \leq m})} 
           \chop \sglt{\neg e}\}}  
  }{
    \infer{
      \semax{e \wedge I \wedge t \leq m}{s_t}
            {(\exists n. \sglt{t = n} \chop \finite \chop \sglt{t < n})
             \chop \sglt{I \wedge t \leq m}}
    }{
      \infer{
      \semax{\exists n.\, e \wedge I \wedge t = n \wedge t \leq m}{s_t}
            {\exists n.\, \sglt{t = n} \chop \finite \chop 
                        \sglt{I \wedge t < n \wedge t \leq m}}
     }{
      \infer[(2)]{
       \forall n.\,
       \semax{e \wedge I \wedge t = n \wedge t \leq m}{s_t}
             {\sglt{t =n} \chop \finite 
                    \chop \sglt{I \wedge t < n \wedge t \leq m}}}{
       \infer*[\forall n.\, \IH_t\, n\, (t = n \wedge t \leq m)]{\forall n.\, 
       \semax{e \wedge I \wedge t = n \wedge t = n \wedge t \leq m}{s_t}
             {\sglt{t = n \wedge t \leq m} \chop \finite 
                    \chop \sglt{I \wedge t < n }}}{
       }
  }}
}}}
\]
\begin{enumerate}
\item The repetition gives rise to a non-empty colist of values of $t$,
which is strictly decreasing and must hence be of finite length. The
concatenation of a finite colist of finite traces is finite.

\item If $t = n$ and $t \leq m$ in the first state of a trace, then $n
\leq m$ (everywhere), so in the last state $t < n$ gives $t < m$,
which can be weakened to $t \leq m$.\qedhere
\end{enumerate}
\end{itemize}
\end{proof}

\begin{cor}\label{prop:total_Hoare}
  For any $s$, $U$ and $Z$, if $\semax{U}{s}{Z}$ is derivable in the
  total-correctness Hoare logic with $\whilefun$, then
  $\semax{U}{s}{\finite \chop \sglt{Z}}$ is derivable in the
  trace-based Hoare logic.
\end{cor}
\begin{proof}
  Immediate from Proposition~\ref{lemma:total_Hoare} by choosing 
  $W = \true$.
\end{proof}

\subsection{Projections of the trace-based logic into standard Hoare logics}

Given that derivations in the standard Hoare logics can be transformed
into derivations in the trace-based Hoare logic, it is natural to
wonder, if it is also possible to translate derivations in the
converse direction. In this direction we would expect some
loss (or displacement) of information. Reducing a condition on traces
into a condition on those last states that happen to exist or into a
condition that also requires their existence must lose or displace the
constraints on the intermediate states. 

We now proceed to demonstrating that meaningful transformations
(``projections'') from the trace-based logic to the standard logics are
indeed possible.

We will show that, if $\semax{U}{s}{P}$ is derivable in the trace-based Hoare
logic, then so are $\semax{U}{s}{\Last\, P}$ in the partial-correctness Hoare
logic and $\semax{U \wedge \ceil{P}{m}}{s}{\Last\, P}$ in the
total-correctness Hoare logic. We will shortly define $\ceil{P}{m}$
formally, but intuitively, the total-correctness judgement
states that $s$ terminates in a state satisfying $\Last\, P$ from any
initial state $\st$ such that $U$ holds in $\st$ and any $\st$-headed
trace satisfying $P$ has length at most $m$.  Since $m$ is universally
quantified on the top level, we can actually derive $\semax{U \wedge
  \exists m.\, \ceil{P}{m}}{s}{\Last\, P}$, stating that $s$
terminates in a state satisfying $\Last\, P$ from any initial state
$\st$ such that $U$ holds and the length of $\st$-headed traces 
satisfying $P$ is bounded.

\begin{prop}\label{lemma:back_Hoare}
For any $s$, $U$ and $P$,
if $\semax{U}{s}{P}$ in the trace-based Hoare logic, 
then for any $W$, 
$\semax{U \wedge W}{s}{\Last{(\sglt{W}\chop P)}}$
is derivable in the partial-correctness Hoare logic.
\end{prop}
\begin{proof}
By induction on the derivation of $\semax{U}{s}{P}$.
We show the main cases of sequence and while.
\begin{itemize}
\item $s = \Seq{s_0}{s_1}$: We are given as the induction hypotheses
  that $\semax{U \wedge W_0}{s_0} {\Last{(\sglt{W_0}
      \chop P \chop \sglt{V})}}$, for any $W_0$, and $\semax{V \wedge
    W_1}{s_1}{\Last{(\sglt{W_1} \chop Q)}}$, for any $W_1$.  We have to show that,
  for any $W$, $\semax{U \wedge W}{\Seq{s_0}{s_1}} {\Last{(\sglt{W}
      \chop P \chop Q)}}$. Let $V'= \Last{(\sglt{W} \chop P \chop \sglt{V})}$.
  We have the derivation
\[
\infer[(1)]{\semax{U \wedge W}{\Seq{s_0}{s_1}}{\Last{(\sglt{W} \chop P \chop Q)}}}{
  \infer{\semax{U \wedge W}{\Seq{s_0}{s_1}}
        {\Last{(\sglt{V'} \chop Q)}}}{
    \infer*[\IH_0\, W]{\semax{U \wedge W}{s_0}{V'}}{
    }
    &   
    \infer[(2)]{\semax{V'}{s_1}
         {\Last{(\sglt{V'} \chop Q)}}}{
       \infer*[\IH_1\, V']{\semax{V \wedge V'}{s_1}
         {\Last{(\sglt{V'} \chop Q)}}}{
       }
    }
  }
}
\]
\begin{enumerate}
\item Recalling that $\Lastnoargs$ is monotone,  we have 
\begin{eqnarray*}
\lefteqn{ \Last{(\sglt{V'} \chop Q)} } \\
& = & 
\Last{(\sglt{\Last{(\sglt{W} \chop P \chop \sglt{V})}} \chop Q)} \\
& \logequ & 
\Last{(\sglt{W} \chop P \chop \sglt{V} \chop Q)} 
\qquad \textrm{(by Lemma~\ref{lemma:misc_asserts} (12))}\\
& \entails & 
\Last{(\sglt{W} \chop P \chop Q)}
\end{eqnarray*}

\item By Lemma~\ref{lemma:misc_asserts} 
(10) and (11),  
$V' = \Last{(\sglt{W} \chop P \chop \sglt{V})} \models V$,
therefore $V' \models V \wedge V'$.
\end{enumerate}

\item $s = \While{e}{s_t}$: We are given as the induction hypothesis that,
for any $W_t$, \linebreak $\semax{I \wedge e \wedge W_t}{s_t}
{\Last{(\sglt{W_t} \chop P \chop \sglt{I})}}$. We now have to prove that, for 
any $W$, 
$\semax{U \wedge W}{\While{e}{s_t}}
{\Last{(\sglt{W} \chop \dup{U} \chop \rep{(P \chop \dup{I})} \chop \sglt{\neg e})}}$,
given $U \models I$. Let $I' =
\Last{(\sglt{W} \chop \dup{U} \chop \rep{(P \chop \dup{I})})}$.
We close the case by the derivation
\[
\infer[(1)]{
\semax{U \wedge W}{\While{e}{s_t}}
{\Last{(\sglt{W} \chop \dup{U} \chop \rep{(P \chop \dup{I})} \chop \sglt{\neg e})}}
}{
  \infer{
    \semax{I'}{\While{e}{s_t}}{I' \wedge \neg e}}{
  \infer[(2)]{
  \semax{e \wedge I'}{s_t}{I'}
  }{
    \infer*[\IH_t\, I']{\semax{e \wedge I \wedge I'}{s_t}
        {\Last{(\sglt{I'} \chop P \chop \sglt{I})}}}{
    }
  }
}
}
\]
\begin{enumerate}
\item We have
\begin{eqnarray*}
\lefteqn{ U \wedge W }\\
& \logequ & 
\Last{(\sglt{W} \chop \dup{U} \chop \sglt{\true})}
\qquad \textrm{(by Lemma~\ref{lemma:misc_asserts} (1) 
and (10))}\\
& \entails &  \Last{(\sglt{W} \chop \dup{U} \chop \rep{(P \chop \dup{I})})}
\qquad \textrm{(by Lemma~\ref{lemma:misc_asserts} (5))}\\
& = & I'
\end{eqnarray*}

\item We have 
\begin{eqnarray*}
\lefteqn{ \Last{(\sglt{I'} \chop P \chop \sglt{I})} }\\
& = & 
\Last{(\sglt{\Last{(\sglt{W} \chop \dup{U} \chop \rep{(P \chop \dup{I})})}} \chop P \chop \sglt{I})}\\
& \logequ & 
\Last{(\sglt{W} \chop \dup{U} \chop \rep{(P \chop \dup{I})} \chop P \chop \sglt{I})}
\qquad \textrm{(by Lemma~\ref{lemma:misc_asserts} (12))}\\
& \logequ & 
\Last{(\sglt{W} \chop \dup{U} \chop \rep{(P \chop \dup{I})} \chop P \chop \dup{I})}\\
& \entails & 
\Last{(\sglt{W} \chop \dup{U} \chop \rep{(P \chop \dup{I})})}\\
& = & I'
\end{eqnarray*}
\end{enumerate}
\end{itemize}\vspace{-18 pt}
\end{proof}\smallskip

\begin{cor}\label{prop:back_Hoare}
For any $s$, $U$ and $P$,
if $\semax{U}{s}{P}$, then 
$\semax{U}{s}{\Last{P}}$ is derivable in the partial correctness Hoare logic.
\end{cor}
\begin{proof}
  Immediate from Proposition~\ref{lemma:back_Hoare} and
  Lemma~\ref{lemma:misc_asserts} (3) by
  instantiating $W = \true$.
\end{proof}

To define the assertion translation for the projection of the
trace-based Hoare logic into the total-correctness Hoare logic, we
introduce two new connectives. The inductively defined trace predicate
$\len{n}$ is true of finite traces with length at most $n$ (we take
the singleton trace to have length 0).  Given a trace predicate $P$,
the state predicate $\ceil{P}{n}$ is defined to be true of a state
$\st$, if every trace headed by $\st$ and satisfying $P$ also
satisfies $\len{n}$.
\[
\begin{array}{c}
\infer{\sglt{\st} \models \len{n}}{}
\quad\quad
\infer{\cons{\st}{\tr} \models \len{(n+1)}}{\tr \models \len{n}}
\\[2ex]
\infer{\st \models \ceil{P}{n}}{
\forall \tr.\, \hd~\tr = \st \wedge \tr \models P \to \tr \models \len{n}}
\end{array}
\]

\begin{lem}\label{lemma:ceil_antimonotone}
For any $n$, $P$ and $Q$, if $Q \models P$ then
$\ceil{P}{n} \models \ceil{Q}{n}$.
\end{lem}

For setoid predicates $P$ and $Q$, we say $Q$ \emph{extends} $P$,
written $\after{P}{Q}$, if, whenever $\tr \models P$,
there exists $\tr'$ such that $\follows{Q}{\tr}{\tr'}$,
i.e., 
$\forall \tr.\, \tr \models P \imp \exists \tr'.\, \follows{Q}{\tr}{\tr'}$.

\begin{lem}\label{lemma:after_monotone}
For any $P, P', Q$ and $Q'$, if $P' \models P$ and $Q \models Q'$
then $\after{P}{Q} \imp \after{P'}{Q'}$. 
\end{lem}
\begin{proof}
Follows from the definition.
\end{proof}

\begin{lem}\label{lemma:after_chop}
For any $P, Q$ and $R$, $\after{(P \chop Q)}{R} \imp
\after{(\sglt{\Last{P}} \chop Q)}{R}$.
\end{lem}
\begin{proof}
Suppose $\tr_0 \models \sglt{\Last{P}} \chop Q$.
There exists $\tr_1$ such that $\tr_1 \models P$
and $\conv{\tr_1}{\hd\ \tr_0}$.
We then have $\tr_1 \concat \tr_0 \models P \chop Q$.
The hypothesis gives us that there exists $\tr_2$ such that
$\follows{Q}{\tr_1 \concat \tr_0}{\tr_2}$.
We then build a trace $\tr_3$ such that
$\follows{Q}{\tr_0}{\tr_3}$ from $\tr_2$ by dropping
the first $n$ elements, where $n$ is the length of $\tr_1$.
($n$ is welldefined because $\tr_1$ is finite.)
\end{proof}

\begin{lem}\label{lemma:after_last}
For any $P$ and $Q$, $\after{\sglt{\Last{P}}}{Q} \imp \after{P}{Q}$.
\end{lem}
\begin{proof}
Suppose $\tr \models P$.
From the hypothesis, we know that, for any $\st$,
if $\conv{\tr}{\st}$, then there exists $\tr'$ such that
$\hd\ \tr' = \st$ and $\tr' \models Q$. 
Using this, we build a trace $\tr'$ such that
$\follows{Q}{\tr}{\tr'}$, by 
traversing $\tr$ and invoking the hypothesis when
the last state of $\tr$ is hit.
\end{proof}

\begin{prop}\label{prop:back_HoareTotal}
  For any $s$, $U$ and $P$, if $\semax{U}{s}{P}$ in the trace-based
  Hoare logic, then for any $m:\mathit{nat}, W$ and $R$ such that
  $\after{\sglt{W} \chop P}{R}$, 
  $\semax{U
    \wedge W \wedge \ceil{P\chop R}{m}}{s} {\Last{(\sglt{W} \chop P)}
    \wedge \ceil{R}{m}}$ is derivable in the total-correctness Hoare
  logic with $\whilerel$.
\end{prop}

\begin{proof}
By induction on the derivation of $\semax{U}{s}{P}$.
We show the main cases of sequence and while.

\begin{itemize}

\item $s = \Seq{s_0}{s_1}$:
We are given as the induction hypotheses that,
for any $m_0, W_0$ and $R_0$ such that
$\after{(\sglt{W_0} \chop P \chop \sglt{V})}{R_0}$, 
we have that
\begin{equation}\tag{$\IH_0$}
\semax{U \wedge W_0 \wedge \ceil{P \chop \sglt{V} \chop R_0}{m_0}}
{s_0}{\Last{(\sglt{W_0} \chop P \chop \sglt{V})} \wedge
\ceil{R_0}{m_0}}
\end{equation}
and, for any $m_1, W_1$ and $R_1$ such that
$\after{(\sglt{W_1} \chop Q)}{R_1}$,
we have that
\begin{equation}\tag{$\IH_1$}
\semax{V \wedge W_1 \wedge \ceil{Q \chop R_1}{m_1}}
{s_1}{\Last{(\sglt{W_1} \chop Q)} \wedge
\ceil{R_1}{m_1}}.
\end{equation}
We have to prove that,
for any $m, W$ and $R$ such that
\begin{equation}\tag{$\Hyp$}
\after{(\sglt{W} \chop P \chop Q)}{R},
\end{equation}
we have that
\begin{equation*}
\semax{U \wedge W \wedge \ceil{P \chop Q \chop R}{m}}{\Seq{s_0}{s_1}}
{\Last{(\sglt{W} \chop P \chop Q)} \wedge \ceil{R}{m}}.
\end{equation*}

Let $V' = \Last{(\sglt{W} \chop P \chop \sglt{V})}$.
We close the case by the derivation
\[
\infer[(5)]{\semax{U \wedge W \wedge \ceil{P \chop Q \chop R}{m}}{\Seq{s_0}{s_1}}
    {\Last{(\sglt{W} \chop P \chop Q)} \wedge \ceil{R}{m}} }{
 \infer{\semax{U \wedge W \wedge \ceil{P \chop Q \chop R}{m}}{\Seq{s_0}{s_1}}
      {\Last{(\sglt{V'} \chop Q)} \wedge \ceil{R}{m}} }{
  \infer[(3)]{\twoline{\{U \wedge W \wedge \ceil{P \chop Q \chop R}{m}\}
          s_0}{\{V' \wedge \ceil{Q \chop R}{m}\}}}{
     \infer*[(1)~\IH_0\, m, W, Q \chop R]{\twoline{\{U \wedge W \wedge \ceil{P \chop \sglt{V} \chop Q \chop R}{m}\}
           s_0}{\{V' \wedge \ceil{Q \chop R}{m}\}}}{
     }
  }
  &
  \infer[(4)]{\twoline{\{
           V' \wedge \ceil{Q \chop R}{m}\}s_1}
          {\{\Last{(\sglt{V'} \chop Q)} \wedge \ceil{R}{m}\}}  }{
     \infer*[(2)~\IH_1\, m, V', R]{\twoline{\{V \wedge V' \wedge \ceil{Q \chop R}{m}\}s_1} 
          {\{\Last{(\sglt{V'} \chop Q)} \wedge \ceil{R}{m}\}} }{  
     }
  }
 }
}
\]
\begin{enumerate}
\item To invoke the induction hypothesis, we have to prove
$\after{\sglt{W} \chop P \chop \sglt{V}}{Q \chop R}$.
By Corollary~\ref{coro:semax_satisfiable} and the hypothesis
$\semax{V}{s_1}{Q}$, we have that,
for any $\st$ such that $\st \models V$, there exists $\tr$ 
such that $\hd~\tr = \st$ and $\tr \models Q$.
This gives us $\after{\sglt{V}}{Q}$.
By Lemmata~\ref{lemma:misc_asserts} (11) 
and~\ref{lemma:after_monotone},
we have $\after{\sglt{\Last{(\sglt{W} \chop P \chop \sglt{V})}}}{Q}$,
therefore $\after{\sglt{W} \chop P \chop \sglt{V}}{Q}$
by Lemma~\ref{lemma:after_last}. 
Now, for any $\tr$, suppose $\tr \models \sglt{W} \chop P \chop \sglt{V}$.
There exists $\tr'$ such that
$\follows{Q}{\tr}{\tr'}$. 
We deduce $\tr' \models \sglt{W} \chop P \chop \sglt{V} \chop Q$.
From the hypothesis $\Hyp$, we know that there exists $\tr''$
such that $\follows{R}{\tr'}{\tr''}$,
therefore $\follows{Q \chop R}{\tr}{\tr''}$, as required.

\item To invoke the induction hypothesis, we have to prove
$\after{\sglt{V'} \chop Q}{R}$, namely
$\after{\sglt{\Last{(\sglt{W} \chop P \chop \sglt{V})}} \chop Q}{R}$.
By Lemma~\ref{lemma:after_chop}, it suffices
to show
$\after{\sglt{W} \chop P \chop \sglt{V} \chop Q}{R}$,
which follows from the hypothesis ($\Hyp$), 
$\sglt{W} \chop P \chop \sglt{V} \chop Q \models
\sglt{W} \chop P \chop Q$ and  Lemma~\ref{lemma:after_monotone}.

\item
$\ceil{P \chop Q \chop R}{m} \models
\ceil{P \chop \sglt{V} \chop Q \chop R}{m}$ follows
from Lemma~\ref{lemma:ceil_antimonotone} and 
$P \chop \sglt{V} \chop Q \chop R \models P \chop Q \chop R$.

\item $V' \models V \wedge V'$ holds because $V' = \Last{(\sglt{W} \chop P
  \chop \sglt{V})} \models \Last{\sglt{V}} \logequ V$
from Lemmata~\ref{lemma:misc_asserts} (10)
and (11).

\item We have
\begin{eqnarray*}
\lefteqn{ \Last{(\sglt{V'} \chop Q)} }\\
& = & 
\Last{(\sglt{\Last{(\sglt{W} \chop P \chop \sglt{V})}} \chop Q)}\\
& \logequ & 
\Last{(\sglt{W} \chop P \chop \sglt{V} \chop Q)}
\qquad \textrm{(by Lemma~\ref{lemma:misc_asserts} (12))}\\
& \entails & 
\Last{(\sglt{W} \chop P \chop Q)}
\end{eqnarray*}
\end{enumerate}

\item $s = \While{e}{s_t}$.
We are given as the induction hypothesis that,
for any $m_t$, $W_t$ and $R_t$ such
that $\after{\sglt{W_t} \chop P \chop \sglt{I}}{R_t}$,
\begin{equation}\tag{$\IH$} 
\semax{e \wedge I \wedge W_t \wedge \ceil{P \chop \sglt{I} \chop R_t}{m_t}}
{s_t}{\Last{(\sglt{W_t} \chop P \chop \sglt{I})} \wedge
\ceil{R_t}{m_t}}. 
\end{equation}
We also know $U \models I$ and $\semax{e \wedge I}{s}{P \chop \sglt{I}}$,
the latter of which gives us
\begin{equation}\label{eq:back_Hoare:whileI}
\semax{I}{\While{e}{s}}{\dup{I} \chop \rep{(P \chop \dup{I})} \chop 
\sglt{\neg e}}.
\end{equation}
We have to prove that, for any $m, W$ and $R$ such
that 
\begin{equation}\tag{$\Hyp$}
\after{\sglt{W} \chop \dup{I} \chop \rep{(P \chop \dup{I})} \chop \sglt{\neg e}}{R}, 
\end{equation}
we have that
\begin{equation*}
\begin{array}{l}
\{U \wedge W \wedge \ceil{\dup{U} \chop \rep{(P \chop \dup{I})}
\chop \sglt{\neg e} \chop R}{m}\}\\
\qquad \While{e}{s_t}~
\{\Last{(\sglt{W} \chop \dup{U} \chop \rep{(P \chop \dup{I})} 
\chop \sglt{\neg e})}
\wedge \ceil{R}{m}\}.
\end{array}
\end{equation*}
Let 
\begin{equation*}
J_0 = \Last{(\dup{U \wedge W} \chop \rep{(P \chop \dup{I})})}, 
\end{equation*}
and
\begin{equation*}
J_1\, n = \ceil{\rep{(P \chop \dup{I})} \chop \sglt{\neg e} \chop R}{n},
\end{equation*}
and 
\begin{equation*}
Q = \dup{I} \chop \rep{(P \chop \dup{I})} \chop \sglt{\neg e} \chop R.
\end{equation*}
We first prove that, for all $n$,
\begin{equation}\label{eq:back_Hoare:while_body}
\semax{e \wedge J_0 \wedge J_1\, n}{s_t}{\exists k.\, k < n \wedge J_0 \wedge J_1\, k}
\end{equation}
holds by case analysis on $n$.
\begin{itemize}
\item $n = 0$.
We prove $e \wedge J_0 \wedge J_1\, 0 \models \false$.
Then, by the derivability of $\semax{\false}{s}{U}$ for any $s$ and $U$
in the total correctness Hoare logic, we obtain 
(\ref{eq:back_Hoare:while_body}).
Suppose $\st_0 \models e \wedge J_0 \wedge J_1\, 0$ holds.
From $\st_0 \models J_0$, there exists 
a trace $\tr_0$ such that
$\tr_0 \models \dup{U \wedge W} \chop \rep{(P \chop \dup{I})}$
and $\conv{\tr_0}{\st_0}$. 
By Lemma~\ref{lemma:misc_asserts} (14),
$\st_0 \models J_0$ and 
$U \entails I$ give us $\st_0 \models I$.
Applying Corollary~\ref{coro:semax_satisfiable}
to (\ref{eq:back_Hoare:whileI}), 
we know that there exists
a trace $\tr_1$ such that $\hd~\tr_1 = \st_0$
and $\cons{\st_0}{\tr_1} 
\models \dup{I} \chop \rep{(P \chop \dup{I})} \chop \sglt{\neg e}$.
Hence, we have 
$\tr_1 \models
\sglt{I} \chop \rep{(P \chop \dup{I})} \chop \sglt{\neg e}$.
We deduce
\[
\tr_0 \concat \tr_1 \models
\sglt{W} \chop \dup{I} \chop \rep{(P \chop \dup{I})} \chop \sglt{\neg e},
\]
where $\concat$ was defined by corecursion by
\[
\sglt{\st} \concat \tr = \tr
\qquad
\cons{\st}{\tr} \concat \tr' = \cons{\st}{(\tr \concat \tr')}.
\]
The assumption $\Hyp$ gives us
a trace $\tr_2$ such that
$\follows{R}{\tr_0 \concat \tr_1}{\tr_2}$. 
Since $\tr_0$ is finite, we find the suffix $\tr_3$
of $\tr_2$ such that $\tr_3 \models
\rep{(P \chop \dup{I})} \chop \sglt{\neg e} \chop R$
by dropping the first $n$ elements from $\tr_2$,
where $n$ is the length of $\tr_0$. 
We also have $\hd~\tr_3 = \st_0$ by construction. 
The trace $\tr_3$ cannot be 
a singleton $\sglt{\st_0}$,
which would imply $\st_0 \models e$ and $\st_0 \models \neg e$.
Contradiction by $\st_0 \models J_1\, 0$.

\item $n \not= 0$. We invoke the induction hypothesis $\IH$
on $n$, $J_0$ and $Q$. To do so, we have to prove
\begin{equation}\label{eq:back_Hoare:IH_cond}
\after{\sglt{J_0} \chop P \chop \sglt{I}}{Q}.
\end{equation}
By Lemma~\ref{lemma:after_chop}, it suffices to prove
$\after
{(\dup{U \wedge W} \chop \rep{(P \chop \dup{I})} \chop P \chop \sglt{I})}{Q}$.
From Corollary~\ref{coro:semax_satisfiable} and 
(\ref{eq:back_Hoare:whileI}), 
we deduce
$\after{\sglt{I}}
{\dup{I} \chop \rep{(P \chop \dup{I})} \chop \sglt{\neg e}}$,
hence
$\after
{\dup{U \wedge W} \chop \rep{(P \chop \dup{I})} \chop P \chop \sglt{I}}
{\dup{I} \chop \rep{(P \chop \dup{I})} \chop \sglt{\neg e}}$
by Lemma~\ref{lemma:after_last}. 
This together with the hypothesis ($\Hyp$)
proves (\ref{eq:back_Hoare:IH_cond}).

We have now proved
\[
\qquad \semax{e \wedge I \wedge J_0 \wedge \ceil{P \chop \sglt{I} \chop Q}{n}}
{s_t}{\Last{(\sglt{J_0} \chop P \chop \sglt{I})} \wedge
\ceil{Q}{n}}. 
\]
We close the case by invoking the consequence rule with 
\begin{eqnarray*}
\lefteqn{ e \wedge J_0 \wedge J_1\, n }\\
& = &  
e \wedge J_0
 \wedge \ceil{\rep{(P \chop \dup{I})} \chop \sglt{\neg e} \chop R}{n}\\
& \entails & 
e \wedge J_0
 \wedge \ceil{P \chop \dup{I} \chop \rep{(P \chop \dup{I})} 
\chop \sglt{\neg e} \chop R}{n}\\
&&\textrm{(by Lemma~\ref{lemma:misc_asserts} (5))}\\
& \logequ & 
e \wedge J_0
\wedge \ceil{P \chop \sglt{I} \chop \dup{I} \chop \rep{(P \chop \dup{I})} 
\chop \sglt{\neg e} \chop R}{n}\\
&&\textrm{(as $\dup{I} \logequ \sglt{I} \chop \dup{I}$)}\\
& = &
e \wedge \Last{(\dup{U \wedge W} \chop \rep{(P \chop \dup{I})})} 
\wedge \ceil{P \chop \sglt{I} \chop Q}{n}\\
& \logequ & 
e \wedge I 
\wedge \Last{(\dup{U \wedge W} \chop \rep{(P \chop \dup{I})})} \wedge 
\ceil{P \chop \sglt{I} \chop Q}{n}\\
&&\textrm{(by Lemma~\ref{lemma:misc_asserts} (14) and 
$U \models I$)}\\
& = & 
e \wedge I \wedge J_0 \wedge \ceil{P \chop \sglt{I} \chop Q}{n},
\end{eqnarray*}
and 
\begin{eqnarray*}
\qquad\lefteqn{ \Last{(\sglt{J_0} \chop P \chop \sglt{I})} \wedge \ceil{Q}{n} }\\
& = & 
\Last{(\sglt{\Last{(\dup{U \wedge W} \chop \rep{(P \chop \dup{I})})}} 
\chop P \chop \sglt{I})} 
\wedge \ceil{Q}{n} \\
& \entails & 
\Last{(\dup{U \wedge W} \chop \rep{(P \chop \dup{I})}
\chop P \chop \sglt{I})} 
\wedge \ceil{Q}{n} \\
&&\textrm{(by Lemma~\ref{lemma:misc_asserts} (12))}\\
&\logequ & 
\Last{(\dup{U \wedge W} \chop \rep{(P \chop \dup{I})}
\chop P \chop \dup{I})} 
\wedge \ceil{Q}{n} 
\qquad \textrm{(1)}\\
& \entails &  
\Last{(\dup{U \wedge W} \chop \rep{(P \chop \dup{I})})}
\wedge \ceil{Q}{n} 
\qquad \textrm{(by Lemma~\ref{lemma:misc_asserts} (5))}\\
& =  & 
J_0 \wedge \ceil{\dup{I} \chop \rep{(P \chop \dup{I})} 
\chop \sglt{\neg e} \chop R}{n}\\
& \entails & 
J_0 \wedge \ceil{\rep{(P \chop \dup{I})} 
\chop \sglt{\neg e} \chop R}{(n-1)} \qquad \textrm{(2)}\\
& = & 
J_0 \wedge J_1\, (n-1)\\
& \entails & 
\exists k.\, k < n \wedge J_0 \wedge J_1\, k
\end{eqnarray*}

(1) Given a trace $\tau$ and a state $\st$ satisfying 
$\dup{U \wedge W} \chop \rep{(P \chop \dup{I})}
\chop P \chop \sglt{I}$ and $\conv{\tau}{\st}$, 
we have $\duplast\ \tau \models
\dup{U \wedge W} \chop \rep{(P \chop \dup{I})}
\chop P \chop \dup{I}$ and $\conv{\duplast\ \tau}{\st}$.

(2) From Lemma~\ref{lemma:misc_asserts} (14) 
and $U \models I$,
we know, for any $\st$, if $\st \models J_0$, namely
$\st \models \Last{(\dup{U \wedge W} \chop \rep{(P \chop \dup{I})})}$,
then $\st \models I$. 
Suppose $\st \models J_0 \wedge \ceil{\dup{I} \chop \rep{(P \chop \dup{I})} 
\chop \sglt{\neg e} \chop R}{n}$.
It suffices to prove that, for any $\tr$, 
if $\hd~\tr = \st$ and 
$\tr \models \rep{(P \chop \dup{I})} 
\chop \sglt{\neg e} \chop R$, then $\tr \models \len{(n-1)}$.
We have $\cons{\st}{\tr} \models
\dup{I} \chop \rep{(P \chop \dup{I})} 
\chop \sglt{\neg e} \chop R$, 
which together with 
$\st \models \ceil{\dup{I} \chop \rep{(P \chop \dup{I})} 
\chop \sglt{\neg e} \chop R}{n}$
yields $\cons{\st}{\tr} \models \len{n}$. This gives us
$\tr \models \len{(n-1)}$.

\end{itemize}

\bigskip

\noindent By the rule for the while-statement by taking $J\, n$ to be
$J_0 \wedge J_1\, n$, we deduce
\[
\semax{J_0 \wedge J_1\, m}{\While{e}{s_t}}
      {\exists k.\, k \leq m \wedge J_0 \wedge J_1\, k \wedge \neg e}
\]
We have
\begin{eqnarray*}
\lefteqn{ U \wedge W \wedge \ceil{\dup{U} \chop \rep{(P \chop \dup{I})}
\chop \sglt{\neg e} \chop R}{m} } \\
& \entails & 
\Last{(\dup{U \wedge W} \chop \rep{(P \chop \dup{I})})}
\wedge \ceil{\rep{(P \chop \dup{I})} \chop \sglt{\neg e} \chop R}{m}\\
& = & 
J_0 \wedge J_1\, m
\end{eqnarray*}
and 
\begin{eqnarray*}
\lefteqn{  \exists k.\, k \leq m \wedge J_0 \wedge J_1\, k \wedge \neg e } \\
& \entails & 
J_0 \wedge J_1\, m \wedge \neg e \quad \quad \textrm{(by monotonicity of $J_1$)} \\
& = &  
\Last{(\dup{U \wedge W} \chop \rep{(P \chop \dup{I})})}
\wedge 
\ceil{\rep{(P \chop \dup{I})} \chop \sglt{\neg e} \chop R}{m} \wedge \neg e\\
& \entails & 
\Last{(\sglt{W} \chop \dup{U} \chop \rep{(P \chop \dup{I})} 
\chop \sglt{\neg e})}
\wedge \ceil{R}{m}
\end{eqnarray*}
So, by consequence, we have
\[
\begin{array}{l}\{U \wedge W \wedge \ceil{\dup{U} \chop \rep{(P \chop \dup{I})}
\chop \sglt{\neg e} \chop R}{m}\}\\
\qquad \While{e}{s_t}~
\{\Last{(\sglt{W} \chop \dup{U} \chop \rep{(P \chop \dup{I})} 
\chop \sglt{\neg e})}
\wedge \ceil{R}{m}\}
\end{array}\]
as we wanted.\qedhere 
\end{itemize}
\end{proof}

\begin{cor}\label{cor:back_HoareTotal}
For any $s$, $U$ and $P$,
if $\semax{U}{s}{P}$ in the trace-based Hoare logic, then
$\semax{U \wedge \exists m.\, \ceil{P}{m}}{s}
{\Last{P}}$
is derivable in the total-correctness Hoare logic with $\whilerel$. 
\end{cor}

\begin{proof}
  From Proposition~\ref{prop:back_HoareTotal} by taking $W = \true$,
  $R = \sglt{\true}$ and then invoking the consequence rule and the
  admissible rule
\[
\infer{\semax{\exists m.\, U}{s}{\exists m.\, Z}}{
  \forall m.\, \semax{U}{s}{Z}
}
\]
  of the total-correctness Hoare logic.
\end{proof}

\subsection{Discussion}\label{sec:Hoare:discussion}

Let us now discuss the transformations from and to the
total-correctness Hoare logic. We used two different rules for while.
In the forward direction (the embedding), we used $\whilefun$, whereas
in the backward direction (the projection), we used $\whilerel$.

With $\whilerel$, the total-correctness Hoare logic is complete,
constructively. In particular, $\whilefun$ is weaker and
straightforwardly derivable from $\whilerel$ by instantiating $J~n = (I
\wedge t = n)$.

In the converse direction, attempts to derive $\whilefun$ from
$\whilerel$ constructively meet several problems. It is natural to
define the invariant $I$ as $I = \exists m.\, J\, m$. The variant $t$
should then be defined as
\[
t = \left\{\begin{array}{l@{\quad}l}
  \textrm{the least $k$ such that $J\, k$} & \textrm{if~} \exists m.\, J\, m \\
  \textrm{0 (or any other natural)} & \textrm{if~} \neg \exists m.\, J\, m
  \end{array}
\right.
\]
This is a sensible classical definition. But constructively, there are
two issues with it. First, in the first case, we must be able to find
the least $k$ such that $J\, k$. We have a bound for this search,
which is given by the witness $m_0$ of $\exists m.\, J\, m$. But in
order to search, we must be able to decide where $J\, k$ or $\neg J\,
k$ for each $k < m_0$, which requires $J$ to be decidable.  Second, in
order to choose between the two cases, we must be able to decide
whether $\exists m.\, J\, m$ or $\neg \exists m.\, J\, m$, which is
generally impossible.

These issues show that $\whilefun$ is not fine-tuned for constructive
reasoning. To formulate a more sensitive rule, we may notice that we
only ever need $t$ in contexts where $I = \exists m.\, J\, m$ can be
assumed to hold. So instead of $I \wedge t = n$ we could use $\Sigma p
: I.\, t\, p = n$ making the variant $t$ depend on the proof of the
invariant. The corresponding modified rule is
\[
\infer[\whilefun']{
  \semax{\Sigma p : I.\,  t\, p = m}{\While{e}{s_t}}{\Sigma p : I.\,  t\, p \leq m \wedge \neg e}
}{
  \forall n:\mathit{nat}
  ~\semax{e \wedge \Sigma p : I.\,  t\, p  = n}{s_t}{\Sigma p : I.\,  t\, p < n}
}
\]

Ideally, we would like to use $\whilerel$ (or some interderivable
rule) both in the transformations to and from the trace-based Hoare
logic. But in the forward direction we have been unable to do it. We
conjecture the reason to be that the state-based total-correctness
triples and trace-based (partial-correctness) triples express
different types of propositions about the underlying evaluation
relations: total-correctness triples are about existence of a final
state satisfying the postcondition, whereas partial-correctness
triples are about all traces produced satisfying the
postcondition. This discrepancy shows itself also in our backward
transformation: we need to invoke the totality of the trace-based
evaluation relation.

\section{Examples}\label{sec:examples}

Propositions~\ref{prop:partial_Hoare} and~\ref{lemma:total_Hoare} show
that our trace-based logic is expressive enough to perform the same
type of analyses that the state-based partial or total correctness
Hoare logics can perform.
However, the expressiveness of our logic goes beyond that of
the partial and the total correctness Hoare logics. 
In this section, we demonstrate this by a series of examples. 
We adopt the usual notational convention that 
any occurrence of a variable in a state predicate
represents the value of the variable in the state,
e.g., a state predicate $x + y = 7$ abbreviates $\lambda \st.\, \st\, x + \st\, y = 7$.

\subsection{Unbounded total search}

Since we work in a constructive underlying logic, 
we can distinguish between termination of a run, $\finite$, 
and nondivergence, $\neg \infinite$.
For instance, any unbounded nonpartial search 
fails to be terminating but is nonetheless nondivergent.

This example is inspired by Markov's principle: $(\neg \forall
n.\, \neg B\, n) \to \exists n. B\, n$ for any decidable
predicate $B$ on natural numbers, i.e., a predicate satisfying
$\forall n.\, B\, n \vee \neg B\, n$. Markov's principle is a
classical tautology, but is not valid constructively.  This implies we
cannot constructively prove a statement $s$ that searches a natural
number $n$ satisfying $B$ by successively checking whether $B ~0, B
~1, B ~2, \ldots$ to be terminating. In other words, we cannot
constructively derive a total correctness judgement for $s$.  The
assumption $\neg \forall n.\ \neg B\, n$ only guarantees that $B$ is not
false everywhere, therefore the search cannot diverge; indeed, we can
constructively prove that $s$ is nondivergent in our logic.

We assume given a decidable predicate $B$ on natural numbers
and an axiom \linebreak {\it B\_noncontradictory}: 
$\neg \forall n.\, \neg B\, n$ stating that 
$B$ is not false everywhere. 
Therefore running the statement 
\[
\mathit{Search} \equiv \Seq{\Assign{x}{0}}{\While{\neg B\, x}{\Assign{x}{x+1}}}
\]
cannot diverge: this would contradict {\it B\_noncontradictory}. 
In Proposition~\ref{prop:Markov}
we prove that any trace produced by running $\mathit{Search}$ is nondivergent
and $B ~x$ holds of the last state. 

We define a family of trace predicates
$\mathit{cofinally\, n}$
coinductively as follows:
\[
\begin{array}{l}
\infer={
  \cons{\st}{\sglt{\st}} \models \cofinally ~n 
}{
  \st ~x = n
  &B ~n
}
\quad
\infer={
  \cons{\st}{\cons{\st}{\tr}} \models \cofinally ~n
}{
  \st ~x = n
  &\neg B ~n
  &\tr \models \cofinally ~(n+1)
}
\end{array}
\]
$\cofinally\, n$ is a setoid predicate.

A crucial observation is that, in the presence of {\it B\_noncontradictory},
$\cofinally ~0$ is stronger than nondivergent: 

\begin{lem}\label{lemma:cofinally_not_infinite}
$\cofinally ~0 \models \neg \infinite$.
\end{lem}

\begin{proof}
  It is sufficient to prove that, for any $\tr$,
  $\satisfy{\tr}{\cofinally ~0}$ and $\satisfy{\tr}{\infinite}$ are
  contradictory.  Suppose there is a trace $\tr$ such that
  $\satisfy{\tr}{\cofinally ~0}$ and $\satisfy{\tr}{\infinite}$. Then
  by induction on $n$ we can show that, for any $n$, there is a trace
  $\tr'$ such that $\satisfy{\tr'}{\cofinally ~n}$ and
  $\satisfy{\tr'}{\infinite}$. But whenever this holds
  for some $\tr'$ and $n$, then $\neg B\, n$. Hence we also have
  $\forall n.\, \neg B\, n$. But this contradicts ${\it
    B\_noncontradictory}$.
\end{proof}

It is straightforward to prove that, 
if $\tr$ satisfies 
$\cofinally~0$, then its last state (if exists) satisfies $B$ at $x$.
(We refer to the program variable $x$ in the statement $\mathit{Search}$.)

\begin{lem}\label{lemma:cofinally_colast}
$\cofinally~0 \models \true \chop \sglt{B~x}$.
\end{lem}
\begin{proof}
We prove the following condition by coinduction,
from which the lemma follows:
for any $n, \tr$, if $\tr \models \cofinally~n$, 
then $\follows{\sglt{B~x}}{\tr}{\tr}$.
\end{proof}

\begin{prop}\label{prop:Markov}
$\semax{\true}{\mathit{Search}}
{(\true \chop \sglt{B ~x}) \wedge \neg \infinite}$.
\end{prop}
\begin{proof}
A sketch of the derivation is given in Figure~\ref{fig:Markov}. (At
several places we have applied the consequence rule silently.)
\begin{enumerate}
\item We use Lemmata~\ref{lemma:cofinally_not_infinite}
and~\ref{lemma:cofinally_colast}.

\item
$x$ is incremented by one in every iteration until
$B$ holds at $x$.
$\dup{x=0} \chop \rep{(\update{(\neg B\, x)}{x}{x+1}\chop \dup{\true})} \chop 
\sglt{B~x} \models \cofinally~0$ follows 
from the definition of $\cofinally$. (It is proved by coinduction.)

\item We take $\true$ as the loop invariant.\qedhere
\end{enumerate}
\end{proof}

\begin{figure}
\[
\begin{array}{c}
\infer[(1)]{
  \semax{\true}{\Seq{\Assign{x}{0}}{\While{\neg B\, x}{\Assign{x}{x+1}}}}
  {(\true \chop \sglt{B ~x}) \wedge \neg \infinite}
}{
  \infer{
    \semax{\true}{\Seq{\Assign{x}{0}}{\While{\neg B\, x}{\Assign{x}{x+1}}}}
          {\update{\true}{x}{0} \chop \cofinally ~0}
  }{
    \semax{\true}{\Assign{x}{0}}{\update{\true}{x}{0}}
    ~~~\infer[(2)]{
       \semax{x=0}{\While{\neg B\, x}{\Assign{x}{x+1}}}
             {\cofinally~0}
    }{
     \detach{-1}{2}{
     \hspace{-4cm}
     \infer[(3)]{
       \twoline{
        \{x=0\}
        ~\While{\neg B\, x}{\Assign{x}{x+1}}}
         {\{\dup{x=0} \chop 
            \rep{(\update{(\REMOVED{x \geq 0 \wedge} \neg B\, x)}{x}{x+1}\chop \dup{\true})} \chop 
            \sglt{\REMOVED{x \geq 0 \wedge} B~x}\}}
      }{
        \semax{\REMOVED{x \geq 0 \wedge} \neg B\, x}
          {\Assign{x}{x+1}}
          {\update{(\REMOVED{x \geq 0 \wedge} \neg B\, x)}{x}{x+1}}
      }
  }
}}}
\end{array}
\]
\caption{Derivation of 
$\semax{\true}{\mathit{Search}}
{(\true \chop \sglt{B ~x}) \wedge \neg \infinite}$}\label{fig:Markov}
\end{figure}

\subsection{Liveness}

As the similarity of our assertion language to
the interval temporal logic suggests, we can specify and prove
liveness properties. 
In Proposition~\ref{prop:liveness}, we prove that
the statement
\[
\Seq{\Assign{x}{0}}{\While{\true}{\Assign{x}{x+1}}}
\]
eventually sets the value of $x$ to $n$ for any $n:\mathit{nat}$
at some point. 

The example is simple but sufficient to demonstrate core techniques
used to prove liveness properties of more practical examples.  For
instance, imagine that assignment to $x$ involves a system call, with
the assigned value as the argument.  It is straightforward to enrich
traces to record such special events, and we can then apply the same
proof technique to prove the statement eventually performs the system
call with $n$ as the argument for any $n$.

For every $n$, we define inductively a trace predicate $\eventually~
n$ stating that a state $\st$ in which the value of $x$ is $n$ is
eventually reached by finitely traversing $\tr$:
\[
\infer{
  \satisfy{\sglt{\st}}{\eventually ~n}
}{
  \st~x = n
}
\quad
\infer{
  \satisfy{\cons{\st}{\tr}}{\eventually ~n}
}{
  \st~x = n
}
\quad
\infer{
  \satisfy{\cons{\st}{\tr}}{\eventually ~n}
}{
  \satisfy{\tr}{\eventually ~n}
}
\]
If $\tr$ satisfies
$\eventually~n$, then it has a finite prefix followed by a state
that maps $x$ to $n$.

\begin{lem}\label{lemma:eventually_finite}
For any $n, \tr$, if $\tr \models \eventually~n$
then $\tr \models \finite \chop \sglt{x=n} \chop \true$.
\end{lem}
\begin{proof}
By induction on the derivation of $\eventually~n$.
\end{proof}

\begin{prop}\label{prop:liveness}
For any $n:\mathit{nat}$,
$\semax{\true}{\Seq{\Assign{x}{0}}{\While{\true}{\Assign{x}{x+1}}}}{\finite ~\chop \sglt{x = n} \chop \true}$.
\end{prop}
\begin{proof} 
A sketch of the derivation is given in Figure~\ref{fig:liveness}.  
\begin{enumerate}
\item We use Lemma~\ref{lemma:eventually_finite}.

\item
$x$ is incremented by one in every iteration, starting from
zero. Each iteration is finite: the length of the trace of each iteration
is invariably
two. $x$  must eventually become $n$ after a finite number of 
iterations for any $n$.

\item We take $\true$ as the invariant.\qedhere
\end{enumerate}
\end{proof}

\begin{figure}
\[
\begin{array}{c}
\infer[(1)]{
  \semax{\true}{\Seq{\Assign{x}{0}}{\While{\true}{\Assign{x}{x+1}}}}{\finite ~\chop \sglt{x = n} \chop \true}
}{
\infer{
  \semax{\true}{\Seq{\Assign{x}{0}}{\While{\true}{\Assign{x}{x+1}}}}{\update{\true}{x}{0} \chop \eventually ~n}
}{
  \semax{\true}{\Assign{x}{0}}{\update{\true}{x}{0}}
  &
  \infer[(2)]{\semax{x=0}{\While{\true}{\Assign{x}{x+1}}}{\eventually~n}
}{
  \detach{-1}{2}{
  \hspace{-4cm}
    \infer[(3)]{
    \twoline{\{x=0\}
             ~\While{\true}{\Assign{x}{x+1}}}
            {\{\dup{x=0} \chop 
               \rep{(\update{\true}{x}{x+1} \chop \dup{\true})} \chop
               \sglt{\false}\}}
    }{
    \semax{\true}{\Assign{x}{x+1}}{\update{\true}{x}{x+1}}
}}}
}}
\end{array}
\]
\caption{Derivation of $\semax{\true}{\Seq{\Assign{x}{0}}{\While{\true}{\Assign{x}{x+1}}}}{\finite ~\chop \sglt{x = n} \chop \true}$}\label{fig:liveness}
\end{figure}

\subsection{Weak trace equivalence}\label{sec:wbism}

The last example is inspired by a notion of weak trace equivalence:
two traces are weakly equivalent if they are bisimilar
by identifying a finite number of consecutive identical states
with a single state.  It is
conceivable that (strong) bisimilarity is too strong for some applications
and one needs weak bisimilarity. 
For instance, we may want to prove that the observable
behavior, such as the colist of i/o events of a potentially diverging run,
is bisimilar to a particular colist of i/o events. 
Then we must be able to collapse a finite number of non-observable 
internal steps. We definitely should not collapse an infinite number 
of internal steps, otherwise we would end up concluding that
a statement performing an i/o operation after a diverging run,
e.g., $\Seq{\While{\true}{\Skip}}{\mathsf{print} ~\mbox{\it ``hello''}}$,
is observably equivalent to a statement immediately performing the same i/o operation,
e.g., $\mathsf{print} ~\mbox{\it ``hello''}$.

In this subsection, we prove that the trace produced by running
the statement
\[
    \While{\true}
           {(\Seq{\Assign{y}{x}}{
            \Seq{(\While{y \not= 0}{\Assign{y}{y-1}})}
                {\Assign{x}{x+1}}})}
\]
is weakly bisimilar to
the ascending sequence of natural numbers $0 :: 1 :: 2 :: 3 :: \ldots$,
by projecting the value of $x$, assuming that $x$ is initially 0.
The statement differs from that of the previous subsection
in that it ``stutters'' for a finite but unbounded number of steps,
i.e., $\While{y \not= 0}{\Assign{y}{y-1}}$,
before the next assignment to $x$ happens. 

This exercise is instructive in that 
we need to formalize weak trace equivalence in our constructive 
underlying logic. 
We do so by supplying an inductive predicate $\red{\tr}{\tr'}$ stating
that $\tr'$ is obtained from $\tr$ 
by dropping finitely many elements from the beginning, until 
the first state with a different value of $x$ is encountered, 
and a coinductive trace predicate
$\upstream ~n$
stating that the given trace is weakly bisimilar to 
the ascending sequence of natural numbers starting at $n$, by projecting
the value of $x$. Formally: 
\[
\infer{
  \red{\cons{\st}{\tau}}{\tau'}
}{
 \st ~x = \hd ~\tau~ x
 &\red{\tau}{\tau'}
}
\quad
\infer{
  \red{\cons{\st}{\tau}}{\tau'}
}{
 \st ~x \not= \hd ~\tau~x
 &\bism{\tau}{\tau'}
}
\]

\[
\infer={
  \satisfy{\cons{\st}{\tau}}{\upstream ~n}
}{
  \st~x = n
  &\red{\cons{\st}{\tr}}{\tr'}
  &\satisfy{\tr'}{\upstream ~(n+1)}
}
\]
These definitions are tailored to our example.  But a more general
weak trace equivalence can be defined similarly~\cite{interact}.  We
note that our formulation is not the only one possible nor the most
elegant. In particular, with a logic supporting nesting of induction
into coinduction as primitive~\cite{DA:mixcoind}, there is no need to
separate the definition into an inductive part, $\red{\tr}{\tr'}$, and
a coinductive part, $\upstream ~n$.  Yet our formulation is amenable
for formalization in Coq.

We also use an auxiliary trace predicate $\stay{x}$ that is true of a
finite trace in which the value of $x$ does not change and is
non-negative at the end and hence everywhere. It is defined
inductively as follows:
\[
\infer{
  \sglt{\st} \models \stay{x}
}{x \geq 0
}
\quad
\infer{
  \cons{\st}{\tr} \models \stay{x}
}{
  \st~x = \hd\, \tr\, x
  &\tr \models \stay{x}
}
\]

\begin{prop}\label{prop:wkbism}
$\semax{x=0}{s}{\upstream ~0}$
where
$s \equiv \While{\true}
           {(\Seq{\Assign{y}{x}}{
            \Seq{(\While{y \not= 0}{\Assign{y}{y-1}})}
                {\Assign{x}{x+1}}})}$.
\end{prop}
\begin{proof}
A sketch of the derivation is given in Figure~\ref{fig:wstream}.

\begin{enumerate}
\item $x$ is incremented by one in every iteration, starting from zero,
and the loop is iterated forever. 
Each iteration is finite, although globally unbounded. 
So $x$ must increase forever.

\item We take $x \geq 0$ as the invariant. 

\item $y$ is decremented by one in every iteration, starting from
a non-negative number. It must become zero after a finite number of iterations.
$x$ does not change. 

\item We take $y \geq 0$ as the invariant.\qedhere
\end{enumerate}

\end{proof}

\begin{figure}
\[
\begin{array}{c}
 \infer[(1)]{
   \semax{x=0}
         {\While{\true}
                {(\Seq{\Assign{y}{x}}
                {\Seq{(\While{y \not= 0}{\Assign{y}{y-1}})}
                     {\Assign{x}{x+1}}})}}
         {\upstream ~0}
 }{
 \infer[(2)]{
  \twoline{\{x = 0\}
        ~\While{\true}
               {(\Seq{\Assign{y}{x}}
               {\Seq{(\While{y \not= 0}{\Assign{y}{y-1}})}
                    {\Assign{x}{x+1}}})}}
        {\{\dup{x=0} \chop 
        \rep{(\stay{x} \chop \update{(x \geq 0)}{x}{x+1} \chop \dup{x \geq 0})} \chop
        \sglt{\false}\}}
}{
\infer{
  \semax{x \geq 0}
        {\Seq{\Assign{y}{x}}
        {\Seq{(\While{y \not= 0}{\Assign{y}{y-1}})}
             {\Assign{x}{x+1}}}}
        {\stay{x} \chop \update{(x \geq 0)}{x}{x+1} \chop \sglt{x \geq 0}}
}{
 \semax{x \geq 0}{\Assign{y}{x}}{\update{(x \geq 0)}{y}{x}}
 ~~~\infer{
     \twoline{\{x \geq 0 \wedge y \geq 0\}
               ~\Seq{(\While{y \not= 0}{\Assign{y}{y-1}})}{\Assign{x}{x+1}}}
             {\{\stay{x} \chop \update{(x \geq 0)}{x}{x+1}\}}
   }{
    \detach{-1}{2}{
    \hspace{-3cm}
    \infer[(3)]{
    \semax{x \geq 0 \wedge y \geq 0}{\While{y \not= 0}{\Assign{y}{y-1}}}
          {\stay{x} }
    }{
    \infer[(4)]{
    \twoline{\{x \geq 0 \wedge y \geq 0\} 
             ~\While{y \not= 0}{\Assign{y}{y-1}}}
            {\{\dup{x \geq 0 \wedge y\geq 0} \chop 
               \rep{(\update{(y > 0)}{y}{y-1} \chop \dup{y \geq 0})} \chop
               \sglt{y = 0}\}}
    }{
     \semax{y \neq 0 \wedge y \geq 0}{\Assign{y}{y-1}}{\update{y > 0}{y}{y-1} \chop \sglt{y \geq 0}}
   }}}
    ~~\semax{x \geq 0}{\Assign{x}{x+1}}{\update{(x \geq 0)}{x}{x+1}}
    }
}}}
\end{array}
\]
\caption{Derivation of $\semax{x = 0}{s}{\upstream ~0}$}
\label{fig:wstream}
\end{figure}


\section{Related work}
\label{sec:related}

Coinductive big-step semantics for nontermination have been considered
by Cousot and Cousot \cite{CC:biiss} and by Leroy and Grall
\cite{Ler:coibso,LG:coibso} (in the context of the CompCert project,
which is a major demonstration of feasibility of certified
compilation). Cousot and Cousot \cite{CC:biiss} study fixpoints on
bi-semantic domains, partitioned into domains of terminating and
diverging behaviors; they prove a specific fixed-point theorem for
such domains (bi-induction), and then produce a bi-inductive big-step
semantics for lambda-calculus (avoiding some duplication of rules
between what would otherwise be distinct inductive and coinductive
definitions of terminating resp.\ diverging evaluation). Leroy and
Grall \cite{LG:coibso} investigate two big-step semantics approaches
for lambda-calculus.  The first, based on Cousot and Cousot
\cite{CC:inddsa}, has different evaluation relations for terminating
and diverging runs, one inductive (with finite traces), the other
coinductive (with infinite traces). To conclude that any program
either terminates or diverges, they need the law of excluded middle
(amounting to decidability of the halting problem), and, as a result,
cannot prove the standard small-step semantics sound wrt.\ the
big-step semantics constructively. The second approach, applied in
\cite{BL:mecscs}, uses a
coinductively defined evaluation relation with possibly infinite
traces. While-loops are not ensured to be progressive in terms of
growing traces (an infinite number of consecutive silent small steps
may be collapsed) and this leads to problems.

Some other works on coinductive big-step semantics include
Glesner~\cite{Gle:procns} and Nestra~\cite{Nes:fras,Nes:trasgf}. In
these it is accepted that a program evaluation can somehow continue
after an infinite number of small steps. With Glesner, this seems to
have been a curious unintended side-effect of the design, which she
was experimenting with just for the interest of it. Nestra developed a
nonstandard semantics with transfinite traces on purpose in order to
obtain a soundness result for a widely used slicing transformation
that is unsound standardly (can turn nonterminating runs into
terminating runs).

Our trace-based coinductive big-step semantics~\cite{NU:trabco} was
heavily inspired by Capretta's~\cite{Cap:genrct} modelling of
nontermination in a constructive setting similar to ours. Rather than
using coinductive possibly infinite traces, he works with a
coinductive notion of a possibly infinitely delayed value (for
statements, this corresponds to delaying the final state). The
categorical basis appears in Rutten's work~\cite{Rut:notcwb}. But
Rutten only studied the classical setting (any program terminates or
not), where a delayed state collapses to a choice of between a state
or a designated token signifying nontermination.

A general categorical account of small-step trace-based semantics has
been given by Hasuo et al.~\cite{HJS:gentsc}.

While Hoare logics for big-step semantics based on inductive, finite
traces have been considered earlier (to reason about traces of
terminating runs), Hoare or VDM-style logics for reasoning about
properties of nonterminating runs do not seem to have been studied
before, with one exception, see below. Neither do we in fact know
about dynamic logic or KAT (Kleene algebra with tests) approaches that
would have assertions about possibly infinite traces.  Rather,
nonterminating runs have been typically reasoned about in temporal
logics like LTL and CTL$^*$ or in interval temporal logic
\cite{Mos:temlrh,HZM:harsbt}.  These are however essentially different
in spirit by their ``endogeneity'': assertions are made about traces
in a fixed transition system rather than traces of runs of different
programs. Notably, however, interval temporal logic has connectives
similar to ours---in fact they were a source of inspiration for our
design.

Hofmann and Pavlova~\cite{HP:ghostvar} consider a VDM-style logic with
finite trace assertions that are applied to all finite prefixes of the
trace of a possibly nonterminating run of a program. This logic
allows reasoning about safety, but not liveness. We expect that we
should be able to embed a logic like this in ours.


\section{Conclusions}
\label{sec:concl}

We have presented a sound and complete Hoare logic for the coinductive
trace-based big-step semantics of While.  The logic naturally extends
both the standard state-based partial and total correctness Hoare
logics. Its design may be exploratory at this stage---in the sense
that one might wish to consider alternative choices of primitive
connectives. We see our logic as a viable unifying foundational
framework facilitating translations from more applied logics.

\subsubsection*{Acknowledgements}

We are grateful to Martin Hofmann, Thierry Coquand and Adam Chlipala
for discussions. Our anonymous reviewers provided very detailed and
useful feedback.


\clearpage

\appendix

\section{State-based partial correctness and total correctness Hoare
  logics}

The figures below give the rules of the standard, state-based partial
correctness and total correctness Hoare logics in the form used in
Section~\ref{sec:Hoare}.

\begin{figure}[h]
\[
\begin{array}{c}
\infer{
  \semax{U[e/x]}{\Assign{x}{e}}{U}
}{}
\quad
\infer{
  \semax{U}{\Skip}{U}
}{}
\quad
\infer{\semax{U}{s_0;s_1}{Z}}{
  \semax{U}{s_0}{V}
  &
  \semax{V}{s_1}{Z}
}
\\[2ex]
\infer{
  \semax{U}{\Ifthenelse{e}{s_t}{s_f}}{Z}
}{
  \semax{e \wedge U}{s_t}{Z}
  &\semax{\neg e \wedge U}{s_f}{Z}
}
\quad
\infer{
  \semax{I}{\While{e}{s_t}}{I \wedge \neg e}
}{
  \semax{e \wedge I}{s_t}{I}
}
\\[2ex]
\infer{
  \semax{\exists z.\, U}{s}{\exists z.\, V}
}{
  \forall z.\, \semax{U}{s}{V}
}
\quad
\infer{
  \semax{U}{s}{Z}
}{
  U \models U'
  &\semax{U'}{s}{Z'}
  & Z' \models Z
}
\end{array}
\]
\caption{Inference rules of partial-correctness Hoare logic}\label{fig:partial_Hoare}
\end{figure}

\begin{figure}[h]
\[
\begin{array}{c}
\infer{
  \semax{U[e/x]}{\Assign{x}{e}}{U}
}{}
\quad
\infer{
  \semax{U}{\Skip}{U}
}{}
\quad
\infer{\semax{U}{s_0;s_1}{Z}}{
  \semax{U}{s_0}{V}
  &
  \semax{V}{s_1}{Z}
}
\\[2ex]
\quad
\infer{
  \semax{U}{\Ifthenelse{e}{s_t}{s_f}}{Z}
}{
  \semax{e \wedge U}{s_t}{Z}
  &\semax{\neg e \wedge U}{s_f}{Z}
}
\\[2ex]
\infer[\whilefun]{
  \semax{I \wedge t = m}{\While{e}{s_t}}{I \wedge t \leq m \wedge \neg e}
}{
  \forall n:\mathit{nat}
  ~\semax{e \wedge I \wedge t = n}{s_t}{I \wedge t < n}
}
\\[2ex]
\mathrm{alt.}
\\[2ex]
\infer[\whilerel]{
  \semax{J\, m}{\While{e}{s_t}}
        {\exists k.\, k \leq m \wedge J\, k \wedge \neg e}
}{
  \forall n:\mathit{nat}
  ~\semax{e \wedge J\, n}{s_t}{\exists k.\, k < n \wedge J\, k}
}

\\[2ex]
\infer{
  \semax{\exists z.\, U}{s}{\exists z.\, V}
}{
  \forall z.\, \semax{U}{s}{V}
}
\quad
\infer{
  \semax{U}{s}{Z}
}{
  U \models U'
  &\semax{U'}{s}{Z'}
  & Z' \models Z
}
\end{array}
\]
\caption{Inference rules of total-correctness Hoare logic}\label{fig:total_Hoare}
\end{figure}

\phantom{blah}
\end{document}